\documentclass[preprint,12pt]{elsarticle}

\usepackage{graphicx}

\usepackage[tight,footnotesize]{subfigure}

\usepackage[cmex10]{amsmath}
\usepackage{amsthm,amssymb}

\usepackage{algorithm}
\usepackage{algorithmic}

\usepackage{array}
\usepackage{mdwmath}
\usepackage{mdwtab}
\usepackage{eqparbox}

\usepackage{url}

\newtheorem{lemma}{Lemma}
\newtheorem{theorem}{Theorem}

\begin{document}

\begin{frontmatter}



\title{Efficient Data Collection in Multimedia Vehicular Sensing Platforms}


\author[raf]{Raffaele Bruno}
\author[madda]{Maddalena Nurchis}

\address[raf]{Institute of Informatics and Telematics (IIT)\\
Italian National Research Council (CNR)\\
Via G. Moruzzi 1, Pisa, ITALY\\
E-mail: r.bruno@iit.cnr.it}

\address[madda]{Converging Networks Laboratory\\
VTT Technical Research Centre of Finland\\
Kaitov$\ddot{a}$yl$\ddot{a}$ 1, 90590 Oulu, Finland\\
E-mail: maddalena.nurchis@gmail.com}

\begin{abstract}
Vehicles provide an ideal platform for urban sensing applications, as they can be equipped with all kinds of sensing devices that can continuously monitor the environment around the travelling vehicle. In this work we are particularly concerned with the use of vehicles as building blocks of a multimedia mobile sensor system able to capture camera snapshots of the streets to support traffic monitoring and urban surveillance tasks. However, cameras are high data-rate sensors while wireless infrastructures used for vehicular communications may face performance constraints. Thus, data redundancy mitigation is of paramount importance in such systems. To address this issue in this paper we exploit sub-modular optimisation techniques to design efficient and robust data collection schemes for multimedia vehicular sensor networks. We also explore an alternative approach for data collection that operates on longer time scales and relies only on localised decisions rather than centralised computations. We use network simulations with realistic vehicular mobility patterns to verify the performance gains of our proposed schemes compared to a baseline solution that ignores data redundancy. Simulation results show that our data collection techniques can ensure a more accurate coverage of the road network while significantly reducing the amount of transferred data. 
\end{abstract}


\end{frontmatter}

%
\section{Introduction\label{sec:introduction}}
\noindent
The idea of using sensors embedded in our vehicles to implement urban monitoring applications is not novel. There are already a few systems, such as Cartel~\cite{HullBZCGMSBM06} or MobiEyes~\cite{LeeMGBC09}, which are designed to support distributed collection and search of sensed data while coping with the intermittent and variable network connectivity of vehicular environments. Most common sensors that are today available on our cars include GPS, vibration sensors, accelerometers, acoustic detectors, proximity sensors, and chemical spill detectors. Thus, typical applications for vehicular sensors networks are road surface monitoring, traffic congestion estimation, and vehicular safety warning services~\cite{LeeG10}. However, a new application paradigm that is emerging for vehicular sensor networks is mobile urban surveillance~\cite{LiHYSLW11,BrunoN13}. The basic idea behind this concept is to take advantage of the vehicles that have forward and/or backward facing cameras to periodically capture camera snapshots (or even short videos) of streets. Then, vehicles can temporarily keep the images recorded by their cameras in a local data storage. The stored images can be delivered to remote data collectors for further processing by opportunistically using 802.11-based roadside access points (APs) that are encountered while travelling~\cite{HullBZCGMSBM06,ErikssonGH08}. Such network of car-based mobile cameras could effectively complement fixed cameras that are currently deployed in our cities (e.g., on traffic lights or public lighting systems) as part of the city traffic control service~\cite{LeeG10}. 

A number of technical challenges should be addressed to effectively utilise a multimedia vehicular sensor network for urban surveillance applications. One of the most important issues is that \emph{the amount of data that can be delivered over the roadside wireless access network is limited} while cameras are high data-rate sensors that require significantly more communication bandwidth than classical sensors (e.g., GPS coordinates). Not only is this due to insufficient capacity of wireless links (e.g., maximum 54~Mbps for IEEE 802.11a/g/p technologies~\cite{std802.11p} but available effective bandwidth is typically much lower~\cite{BychkovskyHMBM06,HadallerKBA07}), but it is also a consequence of the limited number of roadside APs that can be deployed in an urban area for information dissemination. Thus, only intermittent and sporadic connectivity to roadside APs can be typically provided to vehicles in urban areas~\cite{ErikssonBM08}. Furthermore, many images captured by different vehicles that travel along similar routes will suffer from \emph{spatial and temporal correlation that cause data redundancy}. In general, the degree of data redundancy depends on various factors that are difficult to predict or control, such as mobility patterns, available storage space at each vehicle, and network congestion around each roadside AP. Note that due to vehicle mobility it is also difficult to coordinate or schedule monitoring tasks among multiple vehicles. As a consequence, it is of paramount importance to design data collection schemes that minimise the amount of data that needs to be transferred from vehicles to data collectors without: $i)$ degrading the quality of the reconstructed road scene, and $ii)$ violating the latency requirements that are imposed by the urban surveillance applications.

To address the technical issues described above this paper makes the following contributions:
\begin{list}{\tiny$\bullet$}{\leftmargin=1em \itemindent=-0.2em}
\item We define suitable coverage metrics to measure the quality of the road scene reconstruction. Then, we exploit these metrics to formulate the \emph{optimal data collection problem under network capacity constraints as a class of submodular set-covering problems}. Finally, we take advantage of submodularity to develop approximated solutions of the optimisation problems. 
\item We also explore an alternative \emph{probabilistic scheme} that exploits basic aggregate information on the spatio-temporal distribution of received images, and relies only on localised image selection rather than centralised computations to mitigate data redundancy. 
\item We evaluate the performance of our proposed schemes against a naive strategy, in which vehicles attempt to transmit all stored images to the data collector. Results obtained using realistic vehicular mobility patterns in a wide range of different scenarios show that our data collection techniques ensure a more accurate coverage of the road network while significantly reducing the amount of transferred data.
\end{list}
The rest of this paper is organized as follows. Section~\ref{sec:related} outlines related work. In Section~\ref{sec:arch} we describe the architecture of the system considered in this study. In Section~\ref{sec:data_collection} we formulate two submodular optimisation problems for road network coverage and we develop efficient approximations for solving them. Section~\ref{sec:decentralized_solution} presents an alternative probability-based data collection scheme. Section~\ref{sec:performance} describes the simulation environment and explains the simulations results. Section~\ref{sec:conclusions} concludes the paper with final remarks.
%
%
%
\section{Related Work\label{sec:related}}
\noindent
Aspects of this work are related to three main research themes that we discuss in the following. 
%
%
%
%
\subsection{Vehicular sensor networks\label{sec:vsn}}
\noindent
Recently many studies have focused on investigating how to exploit existing sensors that are embedded in vehicles for a plethora of urban sensing applications, and a comprehensive survey on the challenges and opportunities of vehicular sensor networks can be found in~\cite{LeeG10}. To the best of our knowledge one of the first systems that was designed to support data collection from car-based sensors is CarTel~\cite{HullBZCGMSBM06}. CarTel relies on a centralised server, or portal, which is hosted in the wired Internet. Applications running on the portal issue queries about specific sensor data. Then, the CarTel platform dispatches those queries to vehicles and sends replies to the portal while dealing with intermittent connectivity. A somehow similar approach is proposed in~\cite{ErikssonGHNMB08} to specifically monitor the surface conditions of roads by using GPS-enabled smartphones carried by drivers. An alternative system, called MobEyes, is developed in~\cite{LeeMGBC09} to support urban monitoring applications with vehicular sensor networks. MobEyes does not require any fixed infrastructure but it uses mobile agents to opportunistically diffuse sensed data summaries among neighbour vehicles and to create a low-cost index to query monitoring data. Thus, the focus of both CarTel and MobEyes is on distributed opportunistic search of sensed data. 
%
%
%
%
\subsection{Cameras on moving vehicles\label{sec:mvc}}
\noindent
Previous papers have introduced the idea of using cameras in moving vehicles for monitoring tasks~\cite{GreenhillV06,GreenhillV07,LiHYSLW11}. Specifically, both~\cite{GreenhillV06} and~\cite{GreenhillV07} propose to use cameras on buses to assist conventional wide-area video surveillance systems based on fixed cameras. Then, collected images are delivered without interruption to remote data collectors once a day when the buses return to their depot. The focus of those studies is on the development of methods and tools for querying, organising, and transforming images collected from multiple streams acquired from a network of mobile cameras. A distributed and cooperative storage system, called VStore, is developed in~\cite{LiHYSLW11} to support data redundancy elimination and to balance data storage in a multimedia vehicular sensor network used to assist forensic investigations of events, such as traffic accidents. Our work differs from previous studies because we are not concerned with multimedia information processing and storage. On the contrary, we focus on the efficiency and robustness of real-time image data collection when using a roadside wireless infrastructure with limited available bandwidth. For instance, in~\cite{GreenhillV06,GreenhillV07} upload sessions are assumed to last 10 to 12 hours, which allow to transfer hundreds of MB per bus even if wireless communication links are used. Furthermore, most of existing systems that use car-based cameras assume that data collectors operate off-line. On the contrary, in this work we devise an online system in which vehicles continuously deliver the stored data to remote data collectors by opportunistically using the available roadside wireless infrastructure. In this case, upload sessions can last at most a few tens of seconds, and only a few MB of data can be transferred in general due to the harshness of the physical environment~\cite{BychkovskyHMBM06,HadallerKBA07}. Note that preliminary results of this work have been reported in~\cite{BrunoN13}.
%
%
%
%
\subsection{Vehicular Internet access\label{sec:internet}}
\noindent
Various techniques have been proposed to deliver data to and from moving vehicles in a reliable manner using roadside APs. One of the most popular systems is Cabernet~\cite{ErikssonBM08}. Intermittent connectivity between encountered roadside AP causes several challenges, including high connection establishment latency and high loss rates. Cabertnet addresses these issues by using transport protocols and connection establishment procedures that are optimised for the vehicular environment. The Drive-Through Internet~\cite{OttK04}, and Infostations~\cite{SmallH03} projects propose architectures similar to Cabernet, although they are less concerned with network performance issues. More recently, a few studies~\cite{TrullolsFCCB10,ZhengLSK10} have addressed the design of heuristics for the optimal placement of roadside APs with the aim of improving the performance of vehicular Internet access. In this study we simply assume that data delivery occurs when a car associate to one roadside AP encountered during travel and any of the existing systems could be used to improve the performance of vehicle-to-infrastructure communications. 
%
%
%
\section{System Architecture\label{sec:arch}}
\noindent
The system architecture we assume in this study is similar to the one considered in previous works in this field~\cite{LeeMGBC09,LiHYSLW11,GreenhillV06}. More specifically, we assume that $n$ vehicles are moving around in a road network and each vehicle is equipped with: GPS, one forward facing camera, digital maps, and an on-board unit (OBUs) enabling wireless communications with roadside APs, as well as data storage and processing. Then, each vehicle periodically captures images of the road segment ahead using its front camera. In general, the sampling rate of the on-board camera can be variable and tuned according to vehicle movements and application requirements. For instance, the slower is the vehicle and the lower should be the sampling rate to avoid that successive images are covering the same road segment~\cite{GreenhillV07}. Furthermore, the camera has a limited \emph{depth of field}, i.e., the range of distances at which the picture appears acceptably sharp. For the sake of simplicity, we assume that the maximum depth of field is constant and equal to $\delta$~meters. 

Using the GPS receiver the vehicle can determine the location where an image is captured and the time, called timestamp, at which it is recorded. In general, a list of attributes (also tags) can be assigned to each image besides location and time information, such as the identity of the vehicle that generated this image. An essential attribute is also the \emph{validity time} $\tau$, which allows to establish when two images can be regarded as different~\cite{LiHYSLW11}. More precisely, two images can be regarded as same if the difference between their timestamps is smaller than $\tau$. Note that $\tau$ is a system parameter that depends on the application requirements and the variability of the environment to be monitored. For instance, surveillance systems need road images at finer time granularity than classical traffic monitoring systems. Then, a vehicle can avoid storing images that refer to the same road segment if they are equivalent (i.e., their timestamps differ for less than $\tau$~seconds).

In this work we foresee a system in which each vehicle maintains the images that are captured by its mounted camera in a local data storage until it gets connected to a roadside AP (or there is available memory space). Then, the stored images are uploaded using wireless communications to a remote data collector hosted in the wired Internet, in which they can be processed and analysed. In a naive solution each vehicle should try to upload all stored images at the maximum data rate that is allowed by the wireless channel. However, as already observed in Section~\ref{sec:introduction} vehicles traveling similar routes may have many redundant images, and only few of them are needed to reconstruct the road scene. Furthermore, the channel bandwidth is limited and if many vehicles simultaneously try to upload a large number of images congestion will necessarily occur on the wireless medium. Finally, depending on vehicular mobility patterns it is possible that some vehicles have less opportunities to upload their data and they may rapidly utilise their entire storage while other vehicles have a lower utilisation of their buffers. To address those issues we assume that whenever a vehicle associates to a roadside AP it sends to the data collector a \emph{data summary message} with the tags of all its stored images. Then, the data collector is responsible for selecting the minimum set of images that it needs to accurately reconstruct the road scene. It is straightforward to observe that this centralised approach may suffer from scalability issues. Therefore, efficient and robust algorithms are needed to cope with potentially large amounts of data. Furthermore, it is essential to develop algorithms that are able to quickly decide which images should be transferred and which should be discarded because the connection duration between a moving vehicles and a roadside AP is typically short~\cite{BychkovskyHMBM06}.
%
%
%
%
\section{Optimal Data Collection: Problem Formulation and Approximations\label{sec:data_collection}}
\noindent
Ideally we would like to collect the minimum number of images that is needed to reconstruct a road scene with sufficient accuracy in terms of some intuitive metric. To this end, in the following we first introduce the coverage metrics that are used to characterise the quality of the road scene reconstruction. Then, we formulate two distinct optimal data collection problems and we develop approximated solutions of such problems.   
%
%
\subsection{System model\label{sec:system_model}}
\noindent
We model the road network by using the same formalism as in~\cite{ZhengLSK10,TrullolsFCCB10}. Basically a road network is a connected geometric graph $G \!=\! (V,E) $, in which vertices in set $V$ represent the \emph{road intersections} and edges in set $E$ represent the \emph{road segments} connecting road intersections. Without loss of generality we assume straight roads\footnote{The extension to curved road is straightforward as they can be approximated as a sequence of straight lines connecting virtual road intersections~\cite{ZhengLSK10}.}. Furthermore, each edge $e \!\in\! E$ is labeled with a weight $d_{e}$, which represents the physical length of the corresponding road segment. 

Now, let us assume that the remote data collector has received at time $t$ a number of data summaries from the vehicles that get associated to the various roadside APs deployed along the road network. Let $\mathcal{I}_{n}$ the set of all image attributes that are carried in the received summaries\footnote{For notation simplicity variable $t$ is omitted.}. Furthermore, let us denote with $\mathcal{I}_{o}$ the set of all attributes for the images that are stored at the data collector at time $t$. It is important to point out that $\mathcal{I}_{o}$ changes over time because the data collector can decide to delete older images if needed, e.g., to free storage space or, more likely, because the monitoring application sets a delay constraint for image freshness. We now define an intuitive performance metric for quantifying the redundancy degree of the images in set $\mathcal{I}_{n}$ with respect to the images in set $\mathcal{I}_{o}$. This metric is needed to perform an informed decision on which images in $\mathcal{I}_{n}$ are the most valuable for reconstructing the road scene. A straightforward solution to this problem would be to consider which are the road portions covered by each picture and to collect only images than have a minimal overlap with the portions of road segments that are already covered by the images that are stored in the data collector. However, the \emph{spatial distribution of recorded images provides only a partial description of the road coverage problem}. Indeed, the main shortcoming of this approach is that it does not take into account the time validity $\tau$ of each picture as defined in Section~\ref{sec:arch}. In other words, since the road scene may be assumed to change slowly each picture describes a road scene not only at a single time instant but also for a short time interval. Thus, we should also consider the \emph{temporal distribution} of images when computing the quality of the road network coverage. 

\begin{figure}[tb]
\centering
\includegraphics[width=0.5\textwidth,clip=true,angle=-90]{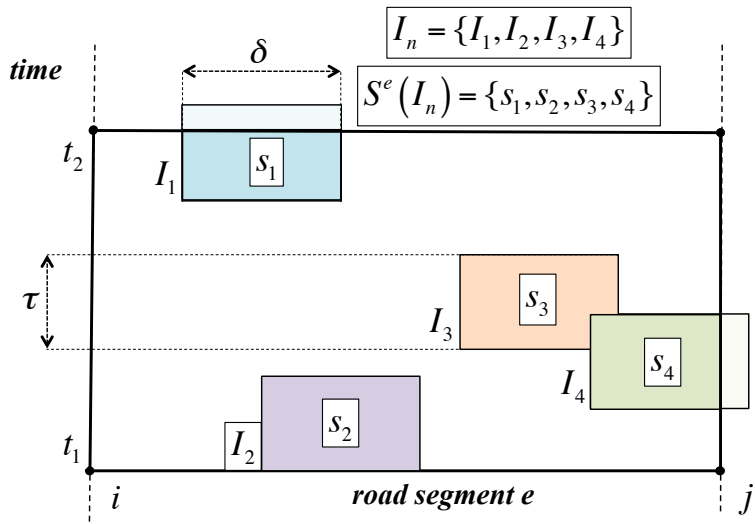}
\caption{Spatio-temporal distribution of images in $\mathcal{I}_{n}$ over road segment $e$ and for time interval $[t_{1},t_{2}]$ (for the sake of figure clarity we assume that $\mathcal{I}_{o} \!=\! \emptyset$). Note that images can cover more than one road segment (as for $I_{4}$) or be valid beyond the time interval $[t_{1},t_{2}]$ (as for $I_{1}$).\label{fig:temporal_dist}}
\vspace{-0.3cm}
\end{figure}
To clarify the above concept, let us assume that four new images are available at time $t$, i.e., $\mathcal{I}_{n} \!=\! \{I_{1},I_{2},I_{3},I_{4}\}$. Then, in Figure~\ref{fig:temporal_dist} we illustrate the spatio-temporal distribution of such images over road segment $e$. Specifically, we associate to each road segment between intersections $i$ and $j$ ($i,j \!\in\! V$) a two-dimensional space in which the $x$-axis represents the distance from $i$ (i.e., $0 \!\le\! x \!\le\! d_{e}$) and the $y$-axis represents the time in the interval $[t_{1},t_{2}]$. In this time~vs.~distance graph an image is characterised by a rectangular area with horizontal side equal to $\delta$ and vertical side equal to $\tau$. The boundaries of such areas partition the two-dimensional space into smaller regions that we call \emph{subareas}. For instance in Figure~\ref{fig:temporal_dist} there are four subareas $\{s_{1},s_{2},s_{3},s_{4}\}$, which are illustrated using different colours. An essential property of subareas is that they do not overlap in time nor in space. More formally, we denote with $\mathcal{S}^{e}(\mathcal{X})$ the set of subareas generated by images in set $\mathcal{X}$ for road segment $e$. For each $s \!\in\! \mathcal{S}^{e}(\mathcal{X})$ we can compute $\theta_{s}$ defined as the size of the correspondent sub-area. Clearly it holds that $\theta_{s} \!\le\! \delta \times \tau$, $\forall s \!\in\! \mathcal{S}^{e}(\mathcal{X})$. 

To finalise the definition of the coverage metric we also need to specify the time horizon over which the network coverage should be measured. To this end we assume that the monitoring application regards as not useful images that are older than a time $T$\footnote{We implicitly assume that the data collector does not store images that are older than $T$~seconds because they do not satisfy the delay requirement set by the monitoring application.}. Note that the system parameter $T$ depends on the specific application requirements, and such latency requirement can be of the order of minutes in non-critical monitoring applications~\cite{LiHYSLW11}\footnote{We can safely assume that $T \!>>\! \tau$.}. Now it is straightforward to define the \emph{road coverage gain} for a road segment $e$ during the time interval $T$ when images in set $\mathcal{X} \! \subseteq \! \mathcal{I}_{n}$ are transferred to the data collector, denoted as $\gamma^{e}(\mathcal{X},\mathcal{I}_{o},T)$, as the fraction of the road segment $e$ that is additionally covered by images in set $\mathcal{X}$. Formally
\begin{equation}
\gamma^{e}(\mathcal{X},\mathcal{I}_{o},T) \!=\! \frac{\sum_{s \in \mathcal{S}^{e}(\mathcal{X} \cup \mathcal{I}_{o})} \theta_{s} - \sum_{s \in \mathcal{S}^{e}(\mathcal{I}_{o})} \theta_{s}}{ d_{e} \times T} \; . \label{eq:gamma_e}
\end{equation}
In other words Equation~(\ref{eq:gamma_e}) quantifies the extent of coverage improvement for a specific road segment if a subset $\mathcal{X}$ of new images is added to the images that are already stored at the data collector. Furthermore, Equation~(\ref{eq:gamma_e}) can be easily generalised to measure the \emph{network coverage gain} for the entire road network, denoted as $\gamma(\mathcal{X},\mathcal{I}_{o},T)$, as follows
\begin{equation}
\gamma(\mathcal{X},\mathcal{I}_{o},T) \!=\! \sum_{e\in E} \gamma^{e}(\mathcal{X},\mathcal{I}_{o},T) \; . \label{eq:gamma}
\end{equation}
%
%
%
%
\subsection{Problem statement\label{sec:problem_statement}}
\noindent
We assume that the sets $\mathcal{I}_{n}$ and $\mathcal{I}_{o}$ are given as part of the input to the data collection decision process. Given a constraint $k$ on the maximum number of images that can be collected, we are looking for a subset $\mathcal{X}_{OPT} \!\subseteq\! \mathcal{I}_{n}$ that maximises the network coverage gain. Formally, the optimisation problem can be expressed as follows: 
\begin{align}
\label{eq:opt1}
\text{\bf OPT1 } \quad & \\ \nonumber
\text{maximize } \quad &\gamma(\mathcal{X},\mathcal{I}_{o},T) \\ \nonumber
\text{subject to:} \quad & w( \mathcal{X}) \le k \; \\ \nonumber
 & \mathcal{X} \!\subseteq\! \mathcal{I}_{n} \; , \nonumber
\end{align}
where $w(\mathcal{X})$ is the cardinality of the set $\mathcal{X}$, i.e., the number of images that are requested from the set $\mathcal{I}_{n}$. It is important to note that constraint $k$ represents the \emph{maximum upload capacity} (measured in number of images) of the system. In the simplest case we could assume that $k$ is set as a fraction of the memory space locally available at each vehicle. For instance, the system may require that only a small fraction of stored images are transferred to the data collector to reduce network capacity utilisation. In a more advanced scenario we can assume that the data collector knows the average number of vehicles that are simultaneously associated to each roadside AP from the logs of those APs. Then, $k$ can be set as a conservative estimate of the maximum number of images that a new vehicle associated to that roadside AP could transmit when sharing the channel bandwidth with other transmitters. 

The problem OPT1 provides a worst-case average guarantee on the quality of the road network coverage. However, depending on the traffic conditions, the patterns of mobile users, and the specific deployment of roadside APs it may occur that specific road segments have a worse coverage than others. In order to improve the coverage fairness among road segments we should ensure that the minimum coverage quality over all road segments is maximised. Formally, the optimisation problem becomes 
\begin{align}
\label{eq:opt2}
\text{\bf OPT2 } \quad & \\ \nonumber
\text{maximize } \quad &\min\limits_{e \in E} \gamma^{e}(\mathcal{X},\mathcal{I}_{o},T) \\ \nonumber
\text{subject to:} \quad & w(\mathcal{X}) \le k \; \\ \nonumber
 & \mathcal{X} \!\subseteq\! \mathcal{I}_{n} \; . \nonumber
\end{align}
It is useful to point out that both OPT1 and OPT2 problems satisfy the same capacity constraint, i.e., less than $k$ images should be transferred at the data collector. From a practical perspective the main difference between the above two optimisation problems consists in the rules used to decide which images to select in the set $\mathcal{X}$. As explained in the following section this has a significant impact on the complexity of each problem. 
%
%
%
\subsection{Approximations via submodular techniques\label{sec:approx}}
\noindent
It is straightforward to observe that problems in~(\ref{eq:opt1}) and~(\ref{eq:opt2}) are variants of classical set covering problems in which we want to ensure that an objective function of the subareas generated by collected images has a maximal value given that we can select a fixed number of images from a given set. It is a well-known result that set covering problems are NP-hard and they do not have polynomial time approximation algorithms under standard assumptions. However, in the following we present an efficient approximation method to solve the above optimisation problems by utilising \emph{submodular optimisation techniques}~\cite{Wolsey82,KrauseMGG08}.

The first step is to demonstrate a few important properties of the set function $\gamma^{e}(\mathcal{X},\mathcal{I}_{o},T)$, which will allow us to develop simple and efficient approximations for the problems formalised in~(\ref{eq:opt1}) and~(\ref{eq:opt2}). The following Lemma summarises these properties.
\begin{lemma}
\label{lem:submodularity}
The set function $\gamma^{e}(\mathcal{A},\mathcal{C},T)$ is:
\begin{enumerate}
\item \emph{normalized}, i.e., $ \gamma^{e}(\emptyset,\mathcal{C},T)$;
\item \emph{nondecreasing}, i.e., 
\begin{equation*}
\gamma^{e}(\mathcal{L},\mathcal{C},T)  \le \gamma^{e}(\mathcal{H},\mathcal{C},T) \textrm{ for all } \mathcal{L}\!\subseteq\! \mathcal{H} \!\subseteq\! \mathcal{A}; 
\end{equation*}
\item \emph{submodular}, i.e.,
\begin{align*}
\gamma^{e}(\mathcal{L} \!\cup\! \{a\},\mathcal{C},T) &\!-\!\gamma^{e}(\mathcal{L},\mathcal{C},T) \!\ge\!  \\
 & \gamma^{e}(\mathcal{H} \!\cup\! \{a\},\mathcal{C},T) \!-\!\gamma^{e}(\mathcal{H},\mathcal{C},T)
\end{align*}
for all $\mathcal{L} \!\subseteq\! \mathcal{H}  \!\subseteq\! \mathcal{A} \textrm{ and } a \!\in\! \mathcal{A}  \!\setminus\! \mathcal{H} ;$ 

\end{enumerate}
\end{lemma}
\begin{proof}
The first property is obvious. The second property holds because if $\mathcal{H}$ is a bigger set than $\mathcal{L}$ then it must necessarily cover an area that is at least equal or greater than the area covered by $\mathcal{L}$. Now, without loss of generality but for notation simplicity let us assume that $\mathcal{C} \!=\! \emptyset$ when proving the third property. It is straightforward to observe that the two terms in the inequality that defines the submodular property can be rewritten as follows 
\begin{equation*}
\gamma^{e}(\mathcal{L} \!\cup\! \{a\},\mathcal{C},T) \!-\!\gamma^{e}(\mathcal{L},\mathcal{C},T) \!=\! \frac{\sum_{s \in \big \{ \mathcal{S}^{e}(a) -  \mathcal{S}^{e}(a) \cap \mathcal{S}^{e}(\mathcal{L}) \big \}} \theta_{s}}{ d_{e} \times T}
\end{equation*}
and that 
\begin{equation*}
\gamma^{e}(\mathcal{H} \!\cup\! \{a\},\mathcal{C},T) \!-\!\gamma^{e}(\mathcal{H},\mathcal{C},T) \!=\! \frac{\sum_{s \in \big \{ \mathcal{S}^{e}(a) -  \mathcal{S}^{e}(a) \cap \mathcal{S}^{e}(\mathcal{H}) \big \}} \theta_{s}}{ d_{e} \times T}
\end{equation*}
Now the proof of the third property reduces to demonstrate that 
\begin{equation}
\mathcal{S}^{e}(a) -  \mathcal{S}^{e}(a) \cap \mathcal{S}^{e}(\mathcal{L}) \supseteq \mathcal{S}^{e}(a) -  \mathcal{S}^{e}(a) \cap \mathcal{S}^{e}(\mathcal{H}) \; ,
\end{equation}
which is equivalent to demonstrate that
\begin{equation}
\label{eq:ineq}
 \mathcal{S}^{e}(a) \cap \mathcal{S}^{e}(\mathcal{L}) \subseteq \mathcal{S}^{e}(a) \cap \mathcal{S}^{e}(\mathcal{H}) \; .
\end{equation}
We prove that inequality~(\ref{eq:ineq}) holds by observing that $ \mathcal{L}\!\subseteq\! \mathcal{H} \!\subseteq\! \mathcal{A}$ implies $ \mathcal{S}^{e}(\mathcal{L}) \!\subseteq\! \mathcal{S}^{e}(\mathcal{H}) \!\subseteq\! \mathcal{S}^{e}(\mathcal{A})$. Therefore, since $\mathcal{S}^{e}(\mathcal{H})$ includes $\mathcal{S}^{e}(\mathcal{L})$ their intersection with the same set $\mathcal{S}^{e}(a)$ will preserve this relationship.  
\end{proof}
As a corollary of Lemma~\ref{lem:submodularity} it is also straightforward to demonstrate the submodularity of the road network coverage $\gamma(\mathcal{X},\mathcal{I}_{o},T)$ since the sum of submodular set functions is also submodular~\cite{NemhauserWF78}. To better explain the physical implications of the submodularity property proven in the above Lemma it is useful to note that submodularity is a type of \emph{diminishing returns property} for a set function~\cite{NemhauserWF78}. Basically, the diminishing returns effect implies that by adding a new image to the set of images already stored at the data collector we increase the quality of the road network coverage more if the data collector has only a few images, and less if it has already collected many images. This behaviour is fundamental to develop simple and efficient approximated solutions of the OPT1 and OPT2 problems. 

First of all, let us start with solving problem OPT1, which only requires the maximisation of a submodular function under a capacity constraint. As shown in~\cite{Wolsey82} the maximisation of a submodular function is amenable to an efficient greedy approximation, which is described in Algorithm~\ref{algo:greedy}. Specifically, Algorithm~\ref{algo:greedy} starts with an empty set and at each iteration it adds a new image taken from set $\mathcal{I}_{n}$ that maximises the normalised increment of the function $\gamma(\mathcal{X},\mathcal{I}_{o},T) $. The algorithm stops when: $i)$ adding a new image does not improve the road network coverage, or $ii)$ the maximum number of images has been selected.  
\begin{algorithm}[!tb]
\small
\caption{Maximum road network coverage}
\label{algo:greedy}
\begin{algorithmic}[1]                
\REQUIRE Input: $\mathcal{I}_{n}$, $\mathcal{I}_{o}$, $T$, $G(V,N)$, $k$,
\ENSURE $\mathcal{X}_{G} \subseteq \mathcal{I}_{n}$ 
\STATE $\mathcal{X}_{G} \leftarrow \emptyset$
\WHILE{$\gamma(\mathcal{X}_{G},\mathcal{I}_{o},T) < \gamma(\mathcal{I}_{n},\mathcal{I}_{o},T) $ \textbf{and} $w(\mathcal{X}_{G})< k$} \label{ln:while}
\STATE Find $a \in \mathcal{I}_{n} \!\setminus\! \mathcal{X}_{G}$ that maximizes $\big \{\gamma(\mathcal{X}_{G} \cup \{a\},\mathcal{I}_{o},T) -  \gamma(\mathcal{X}_{G},\mathcal{I}_{o},T) \big \}$
\STATE  $\mathcal{X}_{G} \leftarrow \mathcal{X}_{G} \cup \{a\}$ \label{ln:endwhile}
\ENDWHILE
\end{algorithmic}
\end{algorithm}
The quality of the approximation provided by a greedy algorithm in the case of submodular set covering problems was investigated in~\cite{NemhauserWF78}, and the following strong approximation guarantees were theoretically established. 
\begin{theorem}
\label{theo:greedy}
Let $\mathcal{C}$ be a finite set and let $F(\mathcal{A})$ be a sub modular set function defined over $\mathcal{C}$. Let $\mathcal{A}_{opt}$ be the solution of the problem $\max_{\mathcal{A} \subseteq \mathcal{C}} \{ F(\mathcal{A}) : |\mathcal{A}|<k \}$. Then, the set $\mathcal{A}_{g}$ obtained by a ``greedy'' search algorithm always produces a solution whose values is at least $(1 - 1/e)$ times the optimal value. Furthermore, there are not other polynomial time algorithms that can ensure a better approximation guarantee.
\end{theorem}
For brevity, hereafter we denote as \textsc{GreedyI} the approximated solution of problem OPT1 that is provided by Algorithm~\ref{algo:greedy}.

Now we show how to solve OPT2, which is more complex than OPT1 since it involves a max-min optimisation. To this end we take advantage of a methodology that was first developed in~\cite{KrauseMGG08} for a similar problem formulation, which relies on an efficient bicriterion approximation that can be achieved by relaxing both the requirement on the objective function and that on the capacity constraint. Specifically, the first step in the solution framework proposed in~\cite{KrauseMGG08} is to solve the following variant of the problem~(\ref{eq:opt2}): given a required minimum road segment coverage $\lambda \!\in\! [0,1]$ over all the road segments, find a set of images in $\mathcal{I}_{n}$ that obtain that with minimum cardinality. Formally, 
\begin{equation}
\min\limits_{\mathcal{X} \subseteq I_{n}} w(\mathcal{X}) \quad \text{subject to} \quad \min\limits_{e \in E}  \gamma^{e}(\mathcal{X},\mathcal{I}_{o},T)  \ge \lambda\; .
\label{eq:lambda}
\end{equation}
A binary search of $\lambda \!\in\! [0,1]$ is then applied. More precisely, for each selected $\lambda$, an instance of problem~(\ref{eq:lambda}) should be solved until a good approximate solution to~(\ref{eq:opt2}) is found. It is easy to show that subproblem~(\ref{eq:lambda}) can be reduced to a submodular set covering problem similar to problem~(\ref{eq:opt1}), which is again amenable to efficient approximation. Specifically, given $\lambda$ we may define the following set function:
\begin{equation}
\phi(\mathcal{X}) = \sum_{e\in E} \min \{\gamma^{e}(\mathcal{X},\mathcal{I}_{o},T),\lambda \} \; .
\label{eq:phi}
\end{equation}
We note that $\phi(\mathcal{X})$ is also a submodular function since: $i)$ $\min \{\gamma^{e}(\mathcal{X},\mathcal{I}_{o},T),\lambda \}$ as a set function on $\mathcal{I}_{n}$ is submodular when $\gamma^{e}(\mathcal{X},\mathcal{I}_{o},T)$ is submodular~\cite{NemhauserWF78} and $ii)$ the sum of submodular functions is submodular. Note that a subset $\mathcal{X} \subseteq \mathcal{I}_{n}$ is a feasible solution to~(\ref{eq:lambda}) if and only if $\phi(\mathcal{X}) \!=\! |\mathcal{I}_{n}| \lambda \!=\! \phi(\mathcal{I}_{n})$ because $i)$ a solution to~(\ref{eq:lambda}) implies that all road segments have a coverage not smaller than $\lambda$, and $ii)$ the coverage obtained with images in $\mathcal{I}_{n}$ cannot be smaller than the one obtained with a subset of $\mathcal{I}_{n}$. Therefore, problem~(\ref{eq:phi}) can be reformulated as follows:
\begin{equation}
\min\limits_{\mathcal{X} \subseteq I_{n}} w(\mathcal{X}) \quad \text{subject to} \quad \phi(\mathcal{X}) = \phi(\mathcal{I}_{n}) \; .
\label{eq:lambda_new}
\end{equation}
Now we can solve OPT2 by solving~(\ref{eq:lambda_new}). Specifically, let $\mathcal{X}(\lambda)$ denote the subset $\mathcal{X} \subseteq \mathcal{I}_{n}$ that achieves $\lambda$ as computed by applying a greedy heuristic similar to Algorithm~\ref{algo:greedy} to~(\ref{eq:lambda_new}). Let us start with $\lambda \!=\! \min_{e \in E} \gamma^{e}(\mathcal{I}_{n},\mathcal{I}_{o},T) $. If $w(\mathcal{X}(\lambda))>k$ a lower $\lambda$ is selected at the next search step\footnote{Note that all road segment with $\gamma^{e}(\mathcal{I}_{n},\mathcal{I}_{o},T) \!=\! 0$ are excluded because they cannot contribute to the road network coverage.}; otherwise, a higher $\lambda$ is selected. This binary search continues until $w(\mathcal{X}(\lambda)) \le k$ and $w(\mathcal{X}(\lambda^{\prime})) > k$ for any $\lambda^{\prime} : \lambda^{\prime} -\lambda \ge \mu$, where $\mu$ can be adjusted to control the accuracy\footnote{To obtain the results shown in Section~\ref{sec:performance} we use $\mu \!=\! 0.01$.}. Given a capacity constraint $k$, such a binary search finds a subset $\mathcal{X} \subseteq \mathcal{I}_{n}$ that has a cardinality lower than $k$ and produces a good approximated solution to OPT2~\cite{KrauseMGG08}. For brevity, hereafter we denote as \textsc{GreedyII} the approximated solution of problem OPT2 that is provided by the binary search described above.
%
%
%
%
\subsection{Practical issues\label{sec:practical}}
\noindent
The assess the feasibility of the \textsc{GreedyI} and \textsc{GreedyII} schemes first of all we should evaluate their computational complexity. If we consider Algorithm~\ref{algo:greedy} it is easy to observe that it requires at most $O(|\mathcal{I}_{n}|)$ iterations (line~\ref{ln:while} to line~\ref{ln:endwhile}) where each iteration involves $|\mathcal{I}_{n}|$ evaluations of $\gamma(\mathcal{X},\mathcal{I}_{o},T)$. On the other hand, the computation of $\gamma(\mathcal{X},\mathcal{I}_{o},T)$ may be quite time consuming for a large road network (i.e., many road segments), and large $|\mathcal{I}_{n}|$ and $|\mathcal{I}_{o}|$. In the case of \textsc{GreedyII} algorithm we should also consider the complexity due to the binary search. More precisely, the greedy heuristic should be recursively executed on a number of different $\lambda$ values, which is equal to $O(\log(\min_{e \in E} \gamma^{e}(\mathcal{I}_{n},\mathcal{I}_{o},T ) / \mu))$. Again this can be time consuming in large systems. Another design choice that may affect the practicality of an optimisation-based approach for addressing the road network coverage is the number of times problem OPT1 (or OPT2) should be solved. In principle, the data collector could invoke the \textsc{GreedyI} (or \textsc{GreedyII}) schemes every time a vehicle gets associated to a roadside AP. In dense urban environments those events may occur quite frequently. Thus, in a real-world system it may be reasonable to set a minimum time interval between two consecutive executions of the greedy algorithms. 

We conclude this section by observing that it could be possible to extend the formulation of the optimisation problem by taking into account additional requirements. For instance, additional information on the local status of each roadside AP (e.g., channel state or bandwidth utilisation) or about the mobility patterns (e.g., popularity of travelled routes) could be leveraged to improve the performance of our algorithm. The downside is that the amount of signalling traffic that needs to be delivered to the remote data collector may increase excessively.  
%
%
%
%
%
%
\section{Centralised versus localised image selection\label{sec:decentralized_solution}}
\noindent
Motivated by the considerations illustrated in Section~\ref{sec:practical}, we now describe a simpler data collection scheme. The main features of this scheme can be summarised as follows: $i)$ it is executed directly on each roadside AP rather than on a single centralised data collector, $ii)$ it operates on longer time scales than \textsc{GreedyI} and \textsc{GreedyII}, and $iii)$ it relies on basic aggregate information on spatio-temporal distributions of received images. Specifically, let us assume that every $S$ seconds the data collector measures, for each road segment $e_{i,j} \!\in\!E$ between intersections $i,j \!\in\! V$, the number $r_{i,j}$ of \emph{redundant images} for that road segment received during that time interval. As an example, two images can be regarded as redundant if more than 50\% of the road scene they cover overlaps (see Figure~\ref{fig:temporal_dist}). Then, the data collector can assign to road segment $e_{i,j}$ a probability $p_{i,j} \!\in\! [0,1]$, defined as the probability of requesting images for that road segment in the following $S$ seconds (how to compute $p_{i,j}$ is explained later on). After that, the data collector distributes the matrix $P \!=\! \{p_{i,j}\}$ to all roadside APs in its area. When a roadside AP receives a data summary from a newly associated vehicle, it uses $p_{i,j}$ values to probabilistically decide which images to request. Intuitively, if in the previous period the data collector has received many redundant images for a road segment $e_{i,j}$ (e.g., because that road segment is in frequently traveled routes) then the probability $p_{i,j}$ should be low. On the contrary, if only few redundant images were received for that road segment then probability $p_{i,j}$ should be high. It is also straightforward to note that network conditions are dynamic (e.g., vehicular traffic profiles change during the day), thus probability $p_{i,j}$ should be continuously adapted. 

\begin{algorithm}[!tbh]
\small
\caption{Controlling the probability $p_{i,j}$ or requesting images for road segment $e_{i,j} \!\in\!E$}
\label{algo:prob}

\begin{algorithmic}[1]                 
\REQUIRE $G(V,N)$, $\delta p$, current $P\!=\! \{p_{i,j}\}$, $r_{H}$, $r_{L}$
\ENSURE updated $P\!=\! \{p_{i,j}\}$ 
\FORALL{$e_{i,j}$ in $E$}
\STATE  $r_{i,j} \leftarrow $ number of received images that covered $e_{i,j}$ in previous $T$ seconds
\IF{ $time=0$}
\STATE $p_{i,j} \leftarrow 1$
\ELSE 
\IF{ $r_{i,j} \ge r_{H}$} 
\STATE $p_{i,j} \leftarrow \max ( \delta p,p_{i,j} - \delta p )$
\ELSIF{ $r_{i,j} \le r_{L}$} 
\STATE $p_{i,j} \leftarrow \min ( 1,p_{i,j} + \delta p )$  
\ENDIF
\ENDIF
\ENDFOR
\end{algorithmic}
\end{algorithm}
Algorithm~\ref{algo:prob} describes how $p_{i,j}$ values are dynamically computed. Initially, all $p_{i,j}$ values are set to one because the data collector has no information on the levels of data redundancy. Then, we implement a \emph{threshold mechanism with hysteresis} to adjust $p_{i,j}$ values at the end of each one of the control period (that last $S$~seconds). More precisely, if $r_{i,j}$ is above a threshold $r_{H}$ we decrease $p_{i,j}$ by a constant factor $0 \!<\! \delta p \!< \! 1$. Note that we do not allow $p_{i,j}$ to be zero because we must continue to update $r_{i,j}$ estimates. Furthermore, to prevent excessive oscillations in $p_{i,j}$ values, which can cause instability, we apply hysteresis in the $p_{i,j}$ adaptation, and we start increasing $p_{i,j}$ only if $r_{i,j}$ is smaller than $r_{L}$, with $r_{L} \!<\! r_{H}$. Following this approach we believe that in some conditions it may be possible to mitigate data redundancy without requiring complex computations. For brevity, the data collection scheme described above is called \textsc{PDC} (or \textsc{Probabilistic Data Collection}).
%
%
%
%
\section{Performance Evaluation\label{sec:performance}}
\noindent
In the following we compare the performance of the proposed data collation schemes versus a naive solution, in which the vehicles are trying to upload all their stored images. Specifically, with the naive approach the vehicles do not sent data summaries to the remote data collector but they try to upload all their images as soon as they get associated with a roadside AP. Thus, the access to the limited channel capacity is regulated only by the MAC protocol and not by the remote data collector.  
%
%
%
%
%
\subsection{Simulation setup\label{sec:setup}}
\noindent
To simulate an urban vehicular network we have utilised VanetMobSim~\cite{HarriFFB09}, a popular generator of realistic vehicular movement traces, and ns-2~\cite{ns2}, a network simulator commonly used to study multi-hop wireless networks. All the following results have been obtained by using a $10 \!\times\! 10$ grid as road network topology, in which each road segment is 100-meter long. Such grid-like road scenario is commonly used in other works in the field to avoid the biases that could be generated by less regular road layouts and non-uniform traffic~\cite{FioreH08,HsiehW12}. Furthermore, we assume that vehicles randomly select their trips in the road network, and they move according to Intelligent Driver Model with Intersection Management (IDM-IM)~\cite{FioreH08}, which is one of the most common car-following models in traffic flow theory~\cite{FrickerW04}. The parameters of this mobility model are chosen as in~\cite{FioreH08}, and they fit real-world urban mobility traces. 

Regarding the surveillance application we assume that each camera has a depth of field ($\delta$) of 30~meters. For simplicity, we also assume that each camera captures an image from the street every two second and that the image validity ($\tau$) is ten seconds\footnote{The maximum speed in urban environments is typically 50~kmph. Thus, a sampling period of two seconds for cameras that have $\delta \!=\! 30$~meters is largely sufficient to ensure continuous coverage of a road as the vehicle travel.}. In addition, we assume that every image captured by camera has size $800 \!\times\! 600 $~pixels, and it is compressed in JPEG format. Thus, one image can be transferred in one packet of size equal to 1000~bytes. Moreover, compared to image size, the image attributes are much smaller. In our work, image summaries are 40 bytes as in~\cite{LiHYSLW11}. Thus, one 1000-byte packet can convey data summaries for 25 images. 

Regarding the communications, images and their attributes are uploaded to the data aggregator using UDP-based connections, while channel access is regulated by the IEEE 802.11 MAC protocol. If not otherwise stated the nominal data rate is 11~Mbps, and the pathloss is modelled using the Two-ray Ground propagation model with a transmission range of 100~meters. Finally, we assume that four roadside APs are regularly deployed within the road network.

In the following we show results aiming at exploring the effectiveness of the proposed data collection schemes under a variety of scenarios. Specifically, we analyse the impact of $i)$ vehicle density, $ii)$ buffer size, and $iii)$ latency requirements on the protocol performance. All simulations have a duration of five hours. A warm-up interval of 30~minutes is used before collecting steady-state statistics to avoid transient effects. For computing 95\% confidence intervals we replicate each simulation five times. 
%
%
\subsection{Impact of vehicle density\label{sec:density}}
\noindent
The following set of results have been obtained under the assumption that: $a)$ each vehicle has a local storage of 200~images, $b)$ the latency requirement ($T$) for the surveillance task is $300$ seconds (i.e., five minutes), and $c)$ \textsc{PDC} is configured with $r_{L} \!=\! 3$, $r_{H} \!=\! 8$, $\delta_{p} \!=\! 0.1$ and $S \!=\! 60$~seconds (the calibration of \textsc{PDC} parameters is discussed in Section~\ref{sec:calibration}). Furthermore, we assume that 200~vehicles are moving in the road network but only a fraction of them is equipped with a camera and it is contributing to the monitoring task. Specifically, we investigate three scenarios in which $5\%$, $12.5\%$ and $25\%$ of the travelling vehicles have a camera, respectively. The effectiveness of the various schemes is assessed in terms of the quality of the road network coverage, the redundancy of collected images, and the bandwidth utilisation of data and signalling traffic. It is important to observe that in our grid-like simulated scenario a road segment (i.e., a road connecting two intersections) is 100 meter long. Therefore, a single image, which has a maximum depth of field of 30~meters can only partially cover a road segment. To capture more accurately this condition we have subdivided each road segment into \emph{subsegments} of fixed length equal to 10 meters (thus, each road segment is composed of ten subsegments). Then, \emph{we compute all statistics per road subsegment} rather than per road segment. This allows us to obtain a more fine grained and precise representation of how many images are actually covering the different portions of each road.

Figure~\ref{fig:accuracy_n} shows the cumulative distribution function (CDF) of the fraction of the total simulation time during which each road subsegment is covered by at least one image for different numbers ($n$) of vehicles equipped with cameras. Our results indicate that the fraction of road subsegments for which the data collector does not receive any image ranges from $60\%$ (with $n \!=\!10$) to $40\%$ (with $n \!=\!50$) if the naive scheme is used. On the contrary, using \textsc{GreedyI} with $k \!=\! 100$ (which implies that up to $50\%$ of locally stored images can be transferred to the remote data collector) this metric ranges from $42\%$ (with $n \!=\! 10$) to $20\%$ (with $n \!=\! 50$). Furthermore, not only does \textsc{GreedyI} reduce the occurrence of ``holes'' in the reconstructed road scene, but it also ensures a more continuous coverage of the road scene. As expected, \textsc{GreedyI} performance degrades by selecting $k \!=\! 10$ because less images can be transferred to the data collector. From the shown results we can observe that \textsc{GreedyII} slightly improves \textsc{GreedyI} for $k \!=\! 100$ in terms of road subsegments that are covered at least by one image. This is mainly due to the fact that \textsc{GreedyII} gives more importance to road segments with poor coverage. However, \textsc{GreedyII} it is also less robust than \textsc{GreedyI} and its performance rapidly degrades when reducing the $k$ value. This can be explained by observing that \textsc{GreedyII} gives higher priority to road segments that are not well covered. Thus, is $k$ is small in some cases most of the selected images can be see to cover a few disadvantaged road segments. With $k \!=\! 100$ this is less likely. Regarding \textsc{PDC} we can observe that it outperforms the naive scheme but it is not as efficient as \textsc{GreedyI}. We remind that \textsc{PDC} is a simple scheme that allows each roadside AP to request with higher probability images related to road segments that are poorly covered. Thus, the main purpose of introducing PDC is to provide preliminary experience on a probabilistic approach to data redundancy elimination.
\begin{figure}[tb]
\centering
\subfigure[$n = 10$]{\includegraphics[width=0.4\textwidth]{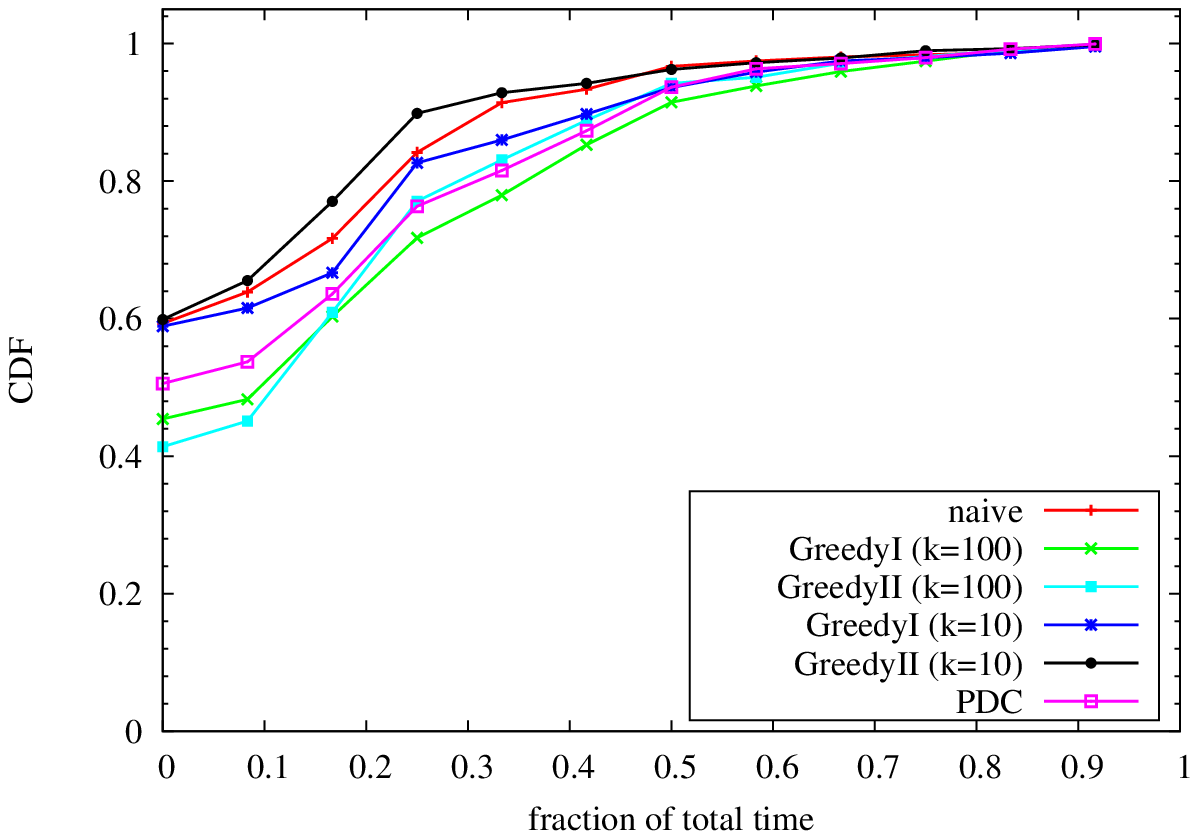}\label{fig:accuracy_n10}}
\subfigure[$n = 25$]{\includegraphics[width=0.4\textwidth]{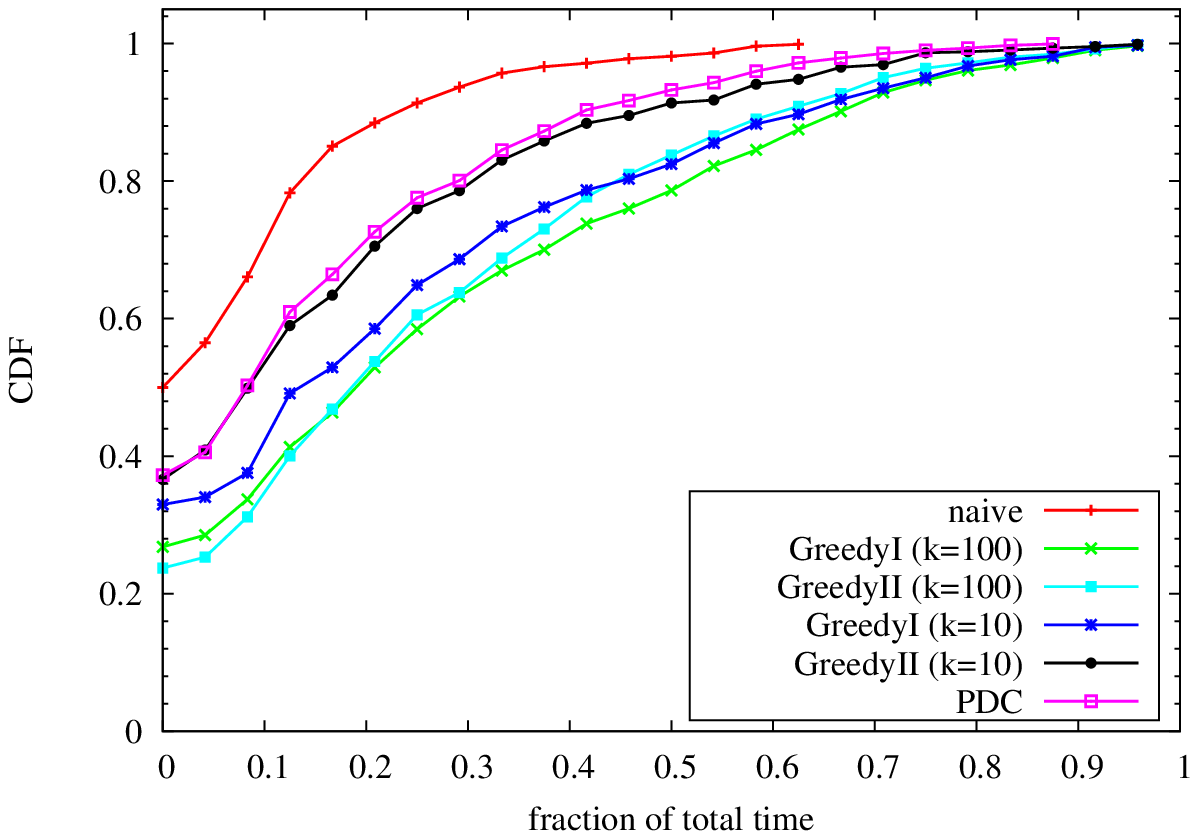}\label{fig:accuracy_n25}}
\subfigure[$n = 50$]{\includegraphics[width=0.4\textwidth]{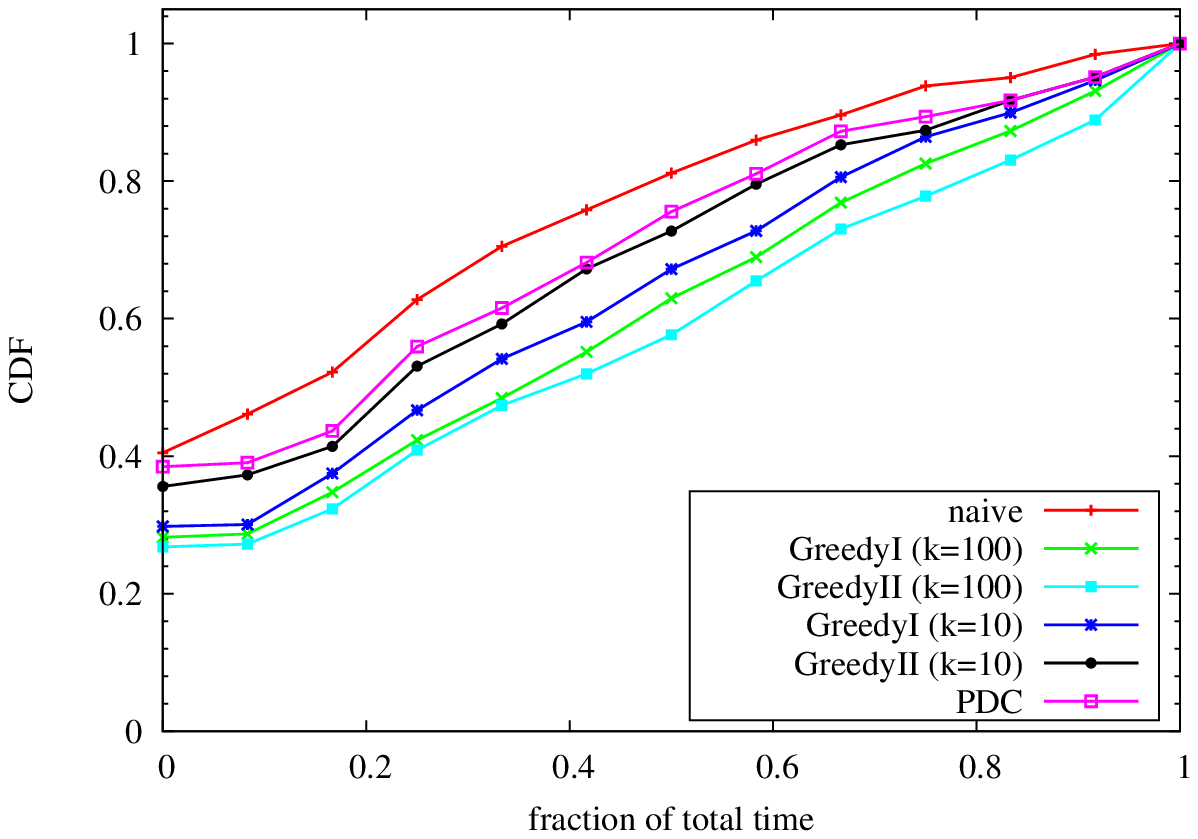}\label{fig:accuracy_n50}}
\caption{CDF of the fraction of total time a road subsegment is covered by at least one image versus the number $n$ of vehicles with cameras.\label{fig:accuracy_n}}
\vspace{-0.2cm}
\end{figure}

It is important to point out that the gain on network road coverage shown in Figure~\ref{fig:accuracy_n} is obtained together with a significant reduction of data redundancy.  To quantify the data redundancy in Figure~\ref{fig:num_cell_n} we show boxplots\footnote{A boxplot is a concise description of a dataset through five values: the bottom and top of the box are the 25$th$ and 90$th$ percentile, respectively, the band in the box is the median (50$th$ percentile) and the ends of the whiskers represent the minimum and maximum values in the dataset.} of the average number of images per road subsegment that are stored at the data collector at any time instant. First of all we can observe that our proposed data collection schemes significantly reduce the maximum data redundancy (up to a factor of four with respect to the naive scheme for $n \!=\! 50$). This is mainly due to the fact that they are able to reduce the number of redundant images that are collected from the most frequently travelled paths. Furthermore, the boxplots indicate that the distribution of the number of collected images is more concentrated towards lower values with our proposed schemes than with the naive approach. Indeed, with our proposed schemes the distribution is generally more concentrated between the median (the band in the box) and the 90$th$ percentile (the top of the box), while with the naive scheme the median is null in most cases (i.e., half of the road subsegments are not covered). Finally, with $k \!=\! 100$ \textsc{GreedyII} provides a slightly higher redundancy than \textsc{GreedyI}. 

It is important to point out that not only is the maximum number of images per road subsegment lower with our schemes than with naive scheme, but we also achieve a \emph{more balanced coverage of the road network}. To clarify this concept, Figure~\ref{fig:redundancy_spatial} plots isoline maps of the spatial distribution of the average number of images per road subsegment over the simulated grid-like road network for $n \!=\! 25$ and different data collection schemes\footnote{Note that all graphs have been obtained by using the same mobility trace.}. Our results indicate that in the naive case there is a high concentration of images at road intersections where image data redundancy is higher because vehicles move slower. On the contrary, with both \textsc{GreedyI} and \textsc{GreedyII} the distribution of the number of images per road subsegment is more uniform and with lower peaks. Qualitatively similar results are obtain with PDC. It is also interesting to observe that the areas with higher values correspond to the places where the four roadside APs are deployed.   
\begin{figure}[tb]
\centering
\includegraphics[width=0.5\textwidth,clip=true,angle=-90]{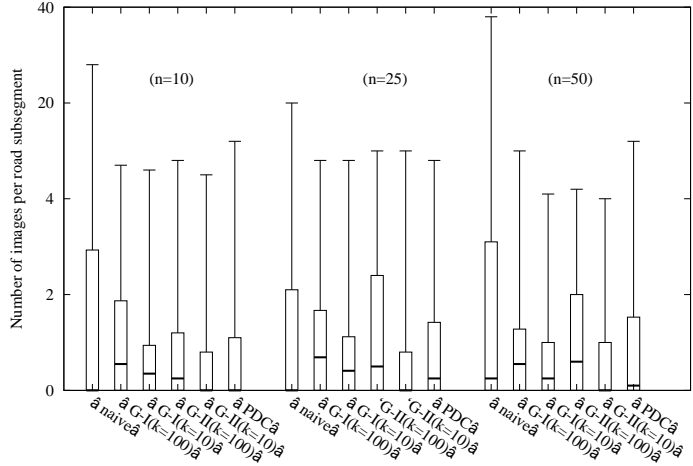}
\caption{Number of images per road subsegment versus the number $n$ of vehicles with cameras}
\label{fig:num_cell_n}
\end{figure}
\begin{figure}[tb]
\centering
\subfigure[naive]{\includegraphics[width=0.4\textwidth]{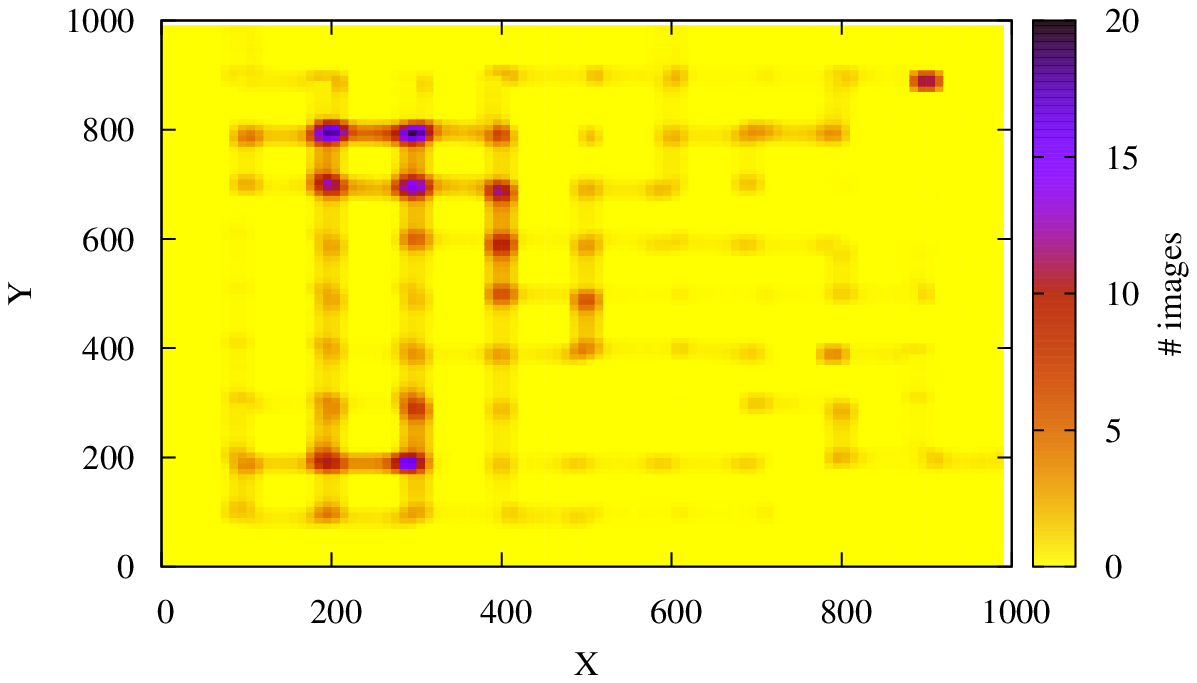}\label{fig:redundancy_base}}
\subfigure[\textsc{GreedyI} ($k=100$)]{\includegraphics[width=0.4\textwidth]{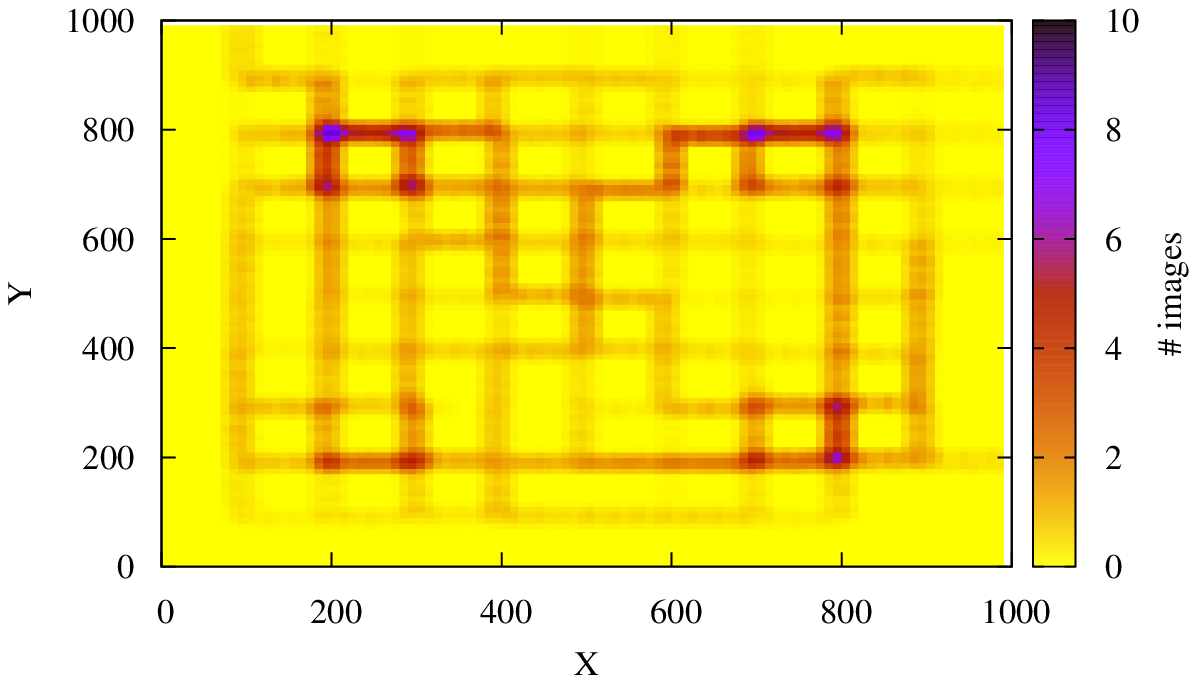}\label{fig:redundancy_greedy}}
\subfigure[\textsc{GreedyII} ($k=100$)]{\includegraphics[width=0.4\textwidth]{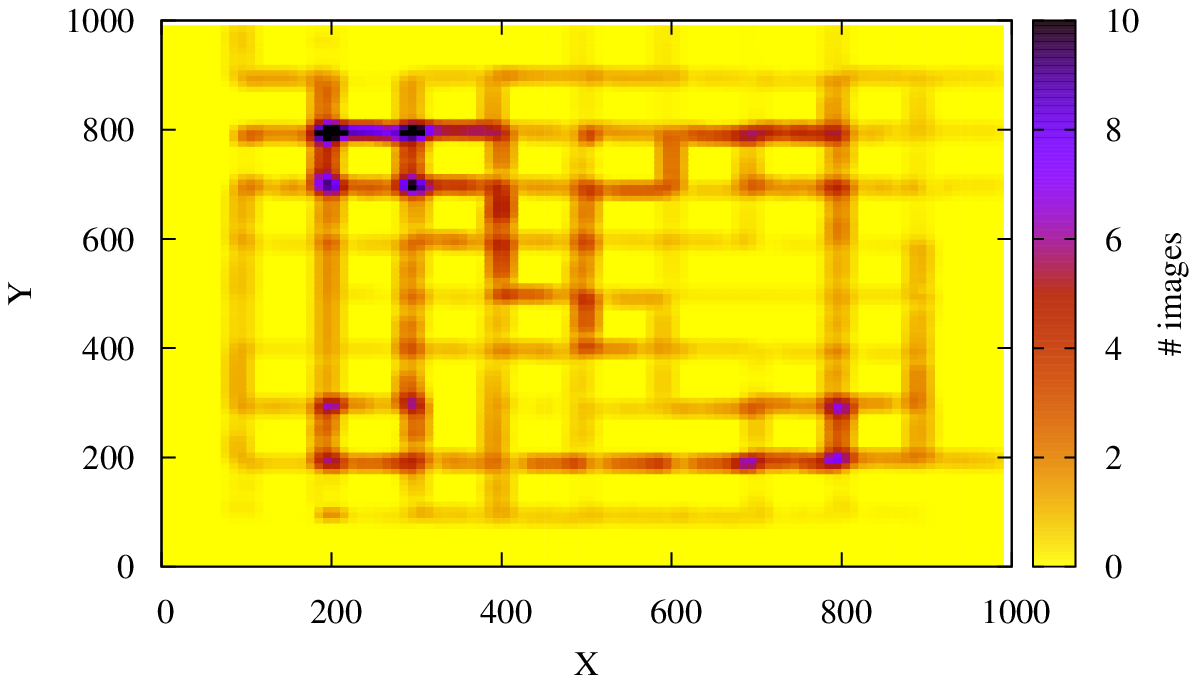}\label{fig:redundancy_maxmim}}
\subfigure[\textsc{PDC}]{\includegraphics[width=0.4\textwidth]{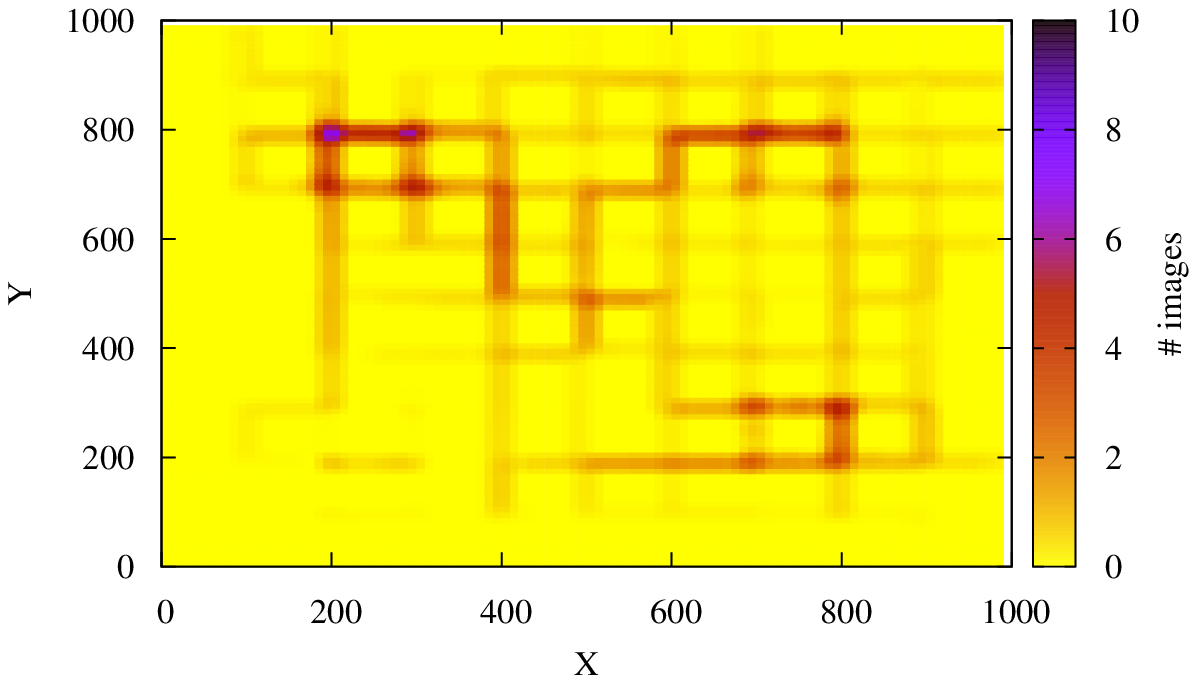}\label{fig:redundancy_pdc}}
\caption{Spatial distribution of the average number of images per road subsegment for $n=25$. \label{fig:redundancy_spatial}}
\vspace{-0.4cm}
\end{figure}

To conclude our performance comparison we evaluate the \emph{efficiency} of each data acquisition scheme in terms of bandwidth consumption. To this end, Figure~\ref{fig:over_data_n} shows the average number of images received at the data collector per minute versus the number $n$ of vehicles equipped with cameras\footnote{We only show the curves for \textsc{GreedyI} since \textsc{GreedyII} obtain similar results.}. As expected the more the vehicles and the higher the number of received images is because associations to roadside APs are more frequent. On the other hand, \textsc{GreedyI} is able to significantly reduce data traffic while improving the coverage of road network because it requests only images that are not redundant. Intuitively, the lower the $k$ value and the lower the number of received images is. Similarly, \textsc{PDC} also reduces the number of transmitted messages, although the extent of this reduction may depend on the thresholds used in PDC to decrease the probability of requesting new images. Finally we also explore protocol overheads (in terms of bytes per minute) due to data summaries and image requests. Specifically, Figure~\ref{fig:over_signalling_n} shows the protocol overheads versus the number $n$ of vehicles equipped with cameras. We remind that an image tag is 40~bytes and data summaries can include tags for at most~200~images (i.e., the size of the local data storage), while the replies list the identifiers of the requested images and they generally consume a few hundreds of bytes at most. Results indicate that all schemes transmit approximately the same amount of data summaries because this value primarily depends on the mobility profiles. On the other hand, the signalling traffic due to image requests decreases with $k$ because this also limits the maximum number of entries into each request.  
\begin{figure}[tb]
\centering
\subfigure[Average number of images received at the data collector per minute\label{fig:over_data_n}]{\includegraphics[width=0.48\textwidth]{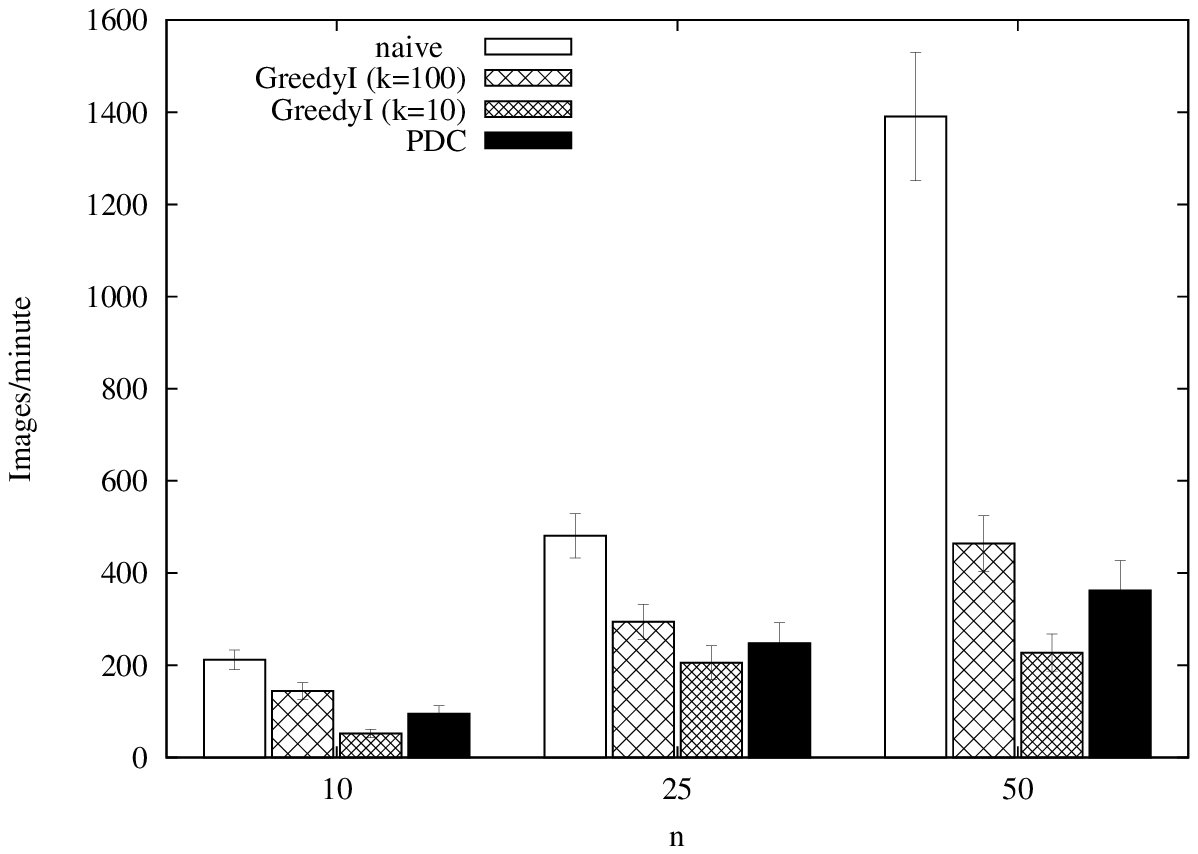}}
\subfigure[Signalling traffic\label{fig:over_signalling_n}]{\includegraphics[width=0.48\textwidth]{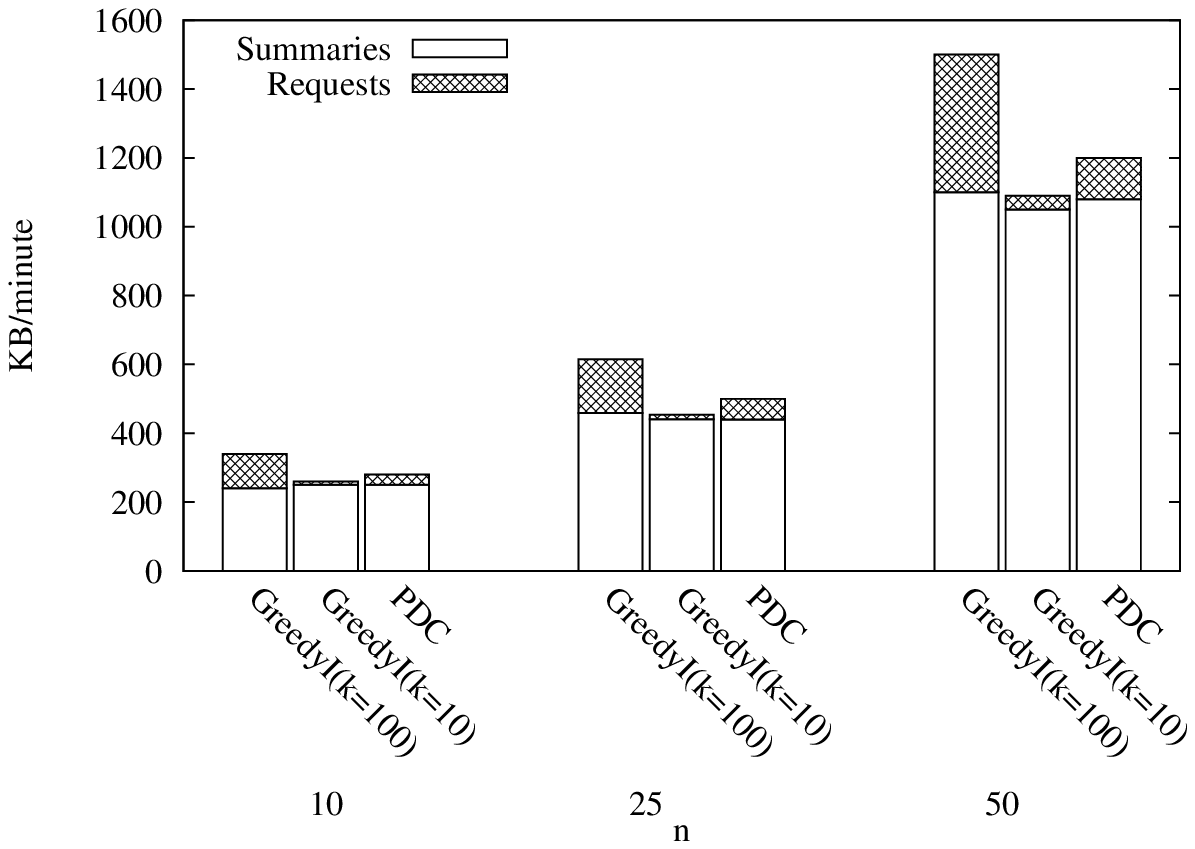}}
\caption{Bandwidth utilisation of signalling and data traffic versus the number $n$ of vehicles equipped with cameras .\label{fig:overhead_n}}
\vspace{-0.3cm}
\end{figure}
%
%
%
%
%
%
\subsection{Impact of buffer size\label{sec:buffer}}
\noindent
It is intuitive to note that a large buffer has the advantage of permitting to keep images also for road segments that are farther from roadside APs. On the other hand the potential gain is limited by the latency constraint imposed by the monitoring task because the oldest images that are stored in the buffer may not satisfy application requirements. In addition, it may be difficult to transfer large buffers due to congestion conditions around the roadside APs. To investigate the interplay between buffer size and system performance in this section we show results for a scenario in which 25~vehicles out of the 200~vehicles that are travelling in the road network are equipped with cameras, and we vary the buffer size. Specifically we investigate two scenarios with buffer equal to 100 and 500 packets respectively (the scenario with buffer equal to 200 is reported in Section~\ref{sec:density}). For the sake of figure clarity we show only results for $k \!=\! 100$ since we have already shown that with $k \!=\! 10$ system performance may degrade excessively (especially for \textsc{GreedyII}).

Figure~\ref{fig:accuracy_buffer} shows the cumulative distribution function (CDF) of the fraction of the total simulation time during which each road subsegment is covered by at least one image for different buffer sizes. Our results indicate that the quality of the reconstructed scene significantly degrades with smaller buffer sizes (i.e., 100 images) for the naive scheme, while larger buffer sizes improve the system performance to some extent. On the other hand both \textsc{GreedyI} and \textsc{GreedyII} are less affected by the buffer size since they always select the $k$ best images that are stored in the buffers. Thus, the main effect of using a larger buffer is to guarantee that the remote controller can select the images to request from a larger set. Finally, \textsc{PDC} outperforms the naive scheme but it is also negatively affected by small buffer sizes, because in this case a more precise selection of useful images is crucial. 
\begin{figure}[tb]
\centering
\subfigure[$B = 100$]{\includegraphics[width=0.4\textwidth]{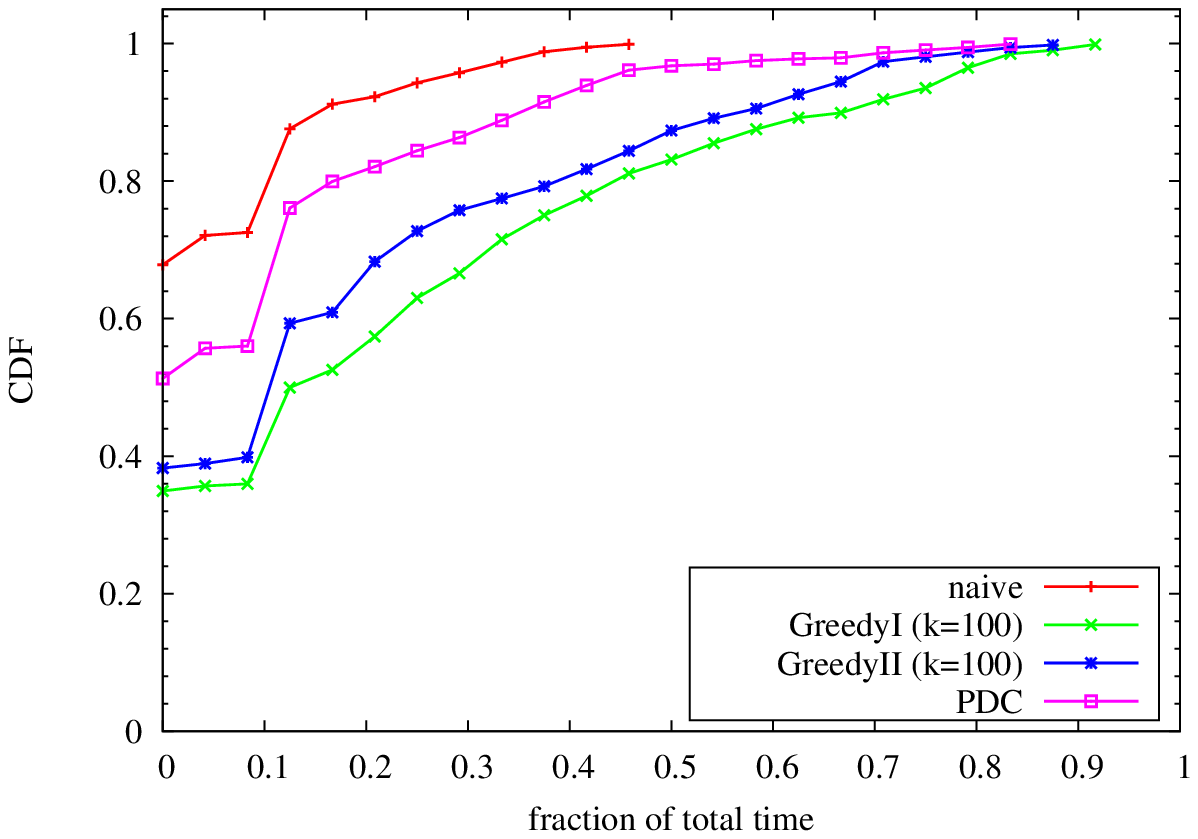}\label{fig:accuracy_buffer100}}
\subfigure[$B = 500$]{\includegraphics[width=0.4\textwidth]{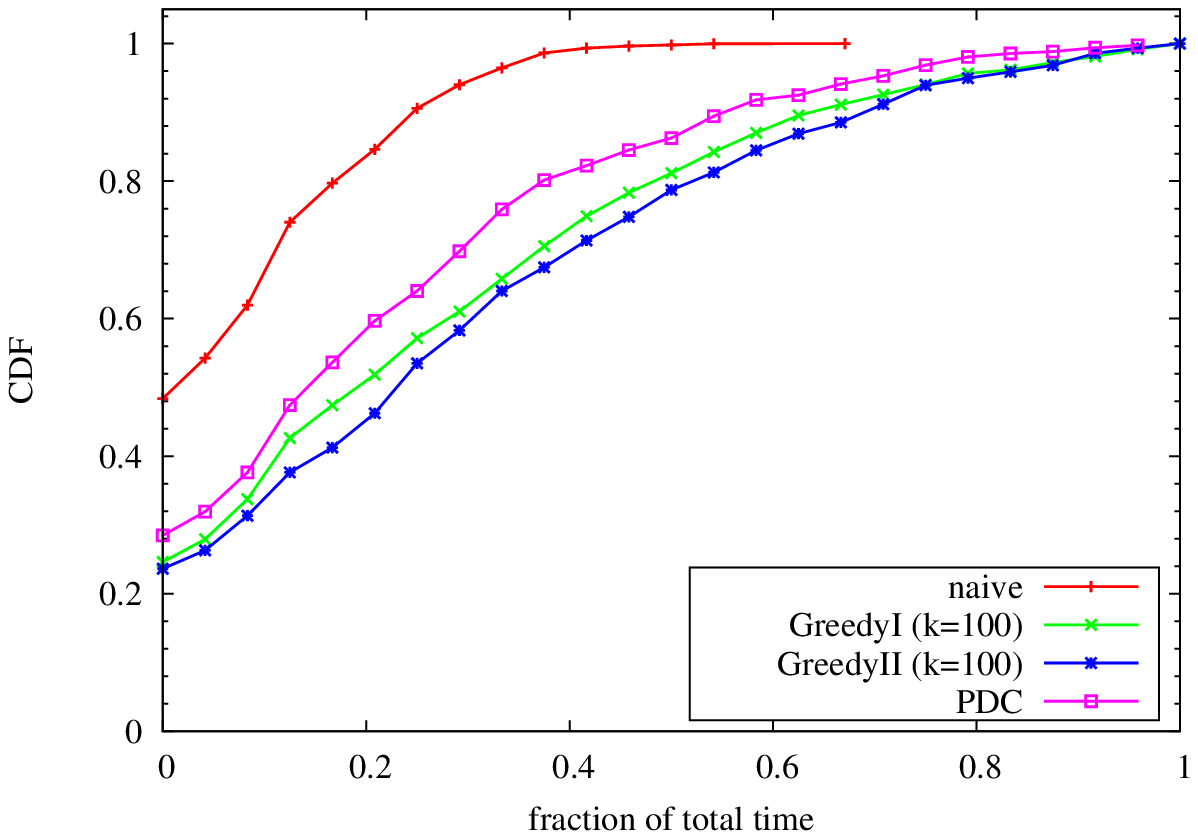}\label{fig:accuracy_buffer500}}
\caption{CDF of the fraction of total time a road subsegment is covered by at least one image versus the buffer size.\label{fig:accuracy_buffer}}
\vspace{-0.2cm}
\end{figure}

In Figure~\ref{fig:num_cell_buffer} we show boxplots of the average number of images per road subsegment that are stored at the data collector at any time instant for different buffer sizes. First we can observe that our proposed schemes significantly reduce the maximum data redundancy. Furthermore, as we have already pointed out in Figure~\ref{fig:num_cell_n} the redundancy distribution is more concentrated between the median (the band in the box) and the 90$th$ percentile (the top of the box), while with the naive scheme the median is always null (i.e., half of the road subsegments are not covered).
\begin{figure}[tb]
\centering
\includegraphics[width=0.5\textwidth,clip=true,angle=-90]{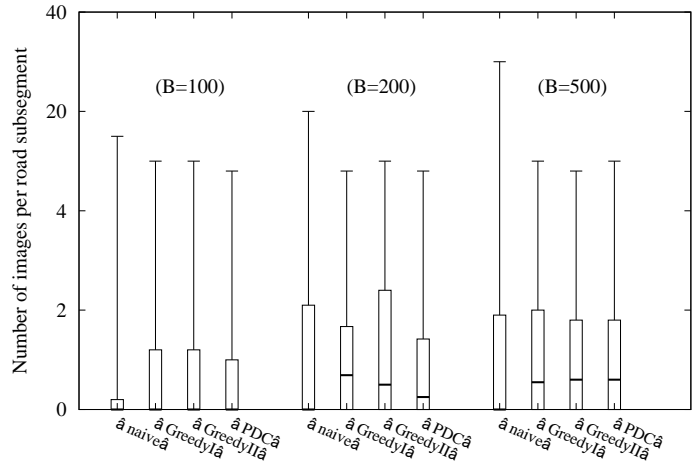}
\caption{Number of images per road subsegment versus the buffer size}
\label{fig:num_cell_buffer}
\end{figure}

Finally, Figure~\ref{fig:overhead_buffer} shows the average number of images received at the data collector per minute for different buffer sizes. Interestingly we observe that the two greedy variants send almost the same number of images, and such number slightly increases with the buffer size. On the other hand, with the naive scheme the number of transferred images increases rapidly with the buffer size, while \textsc{PDC} is inefficient mainly for large buffer sizes. We also explore the protocol overheads in terms of sent data summaries and image requests in Figure~\ref{fig:over_signalling_buffer}. Clearly the plots indicate that all schemes transmit approximately the same amount of data summaries because this parameter mainly depends on the mobility profiles and the buffer occupation. Similarly, the amount of signalling traffic due to image requests is bounded by the parameter $k$ and it slightly decreases with the buffer size.  

\begin{figure}[tb]
\centering
\subfigure[Average number of images received at the data collector per minute\label{fig:over_data_buffer}]{\includegraphics[width=0.48\textwidth]{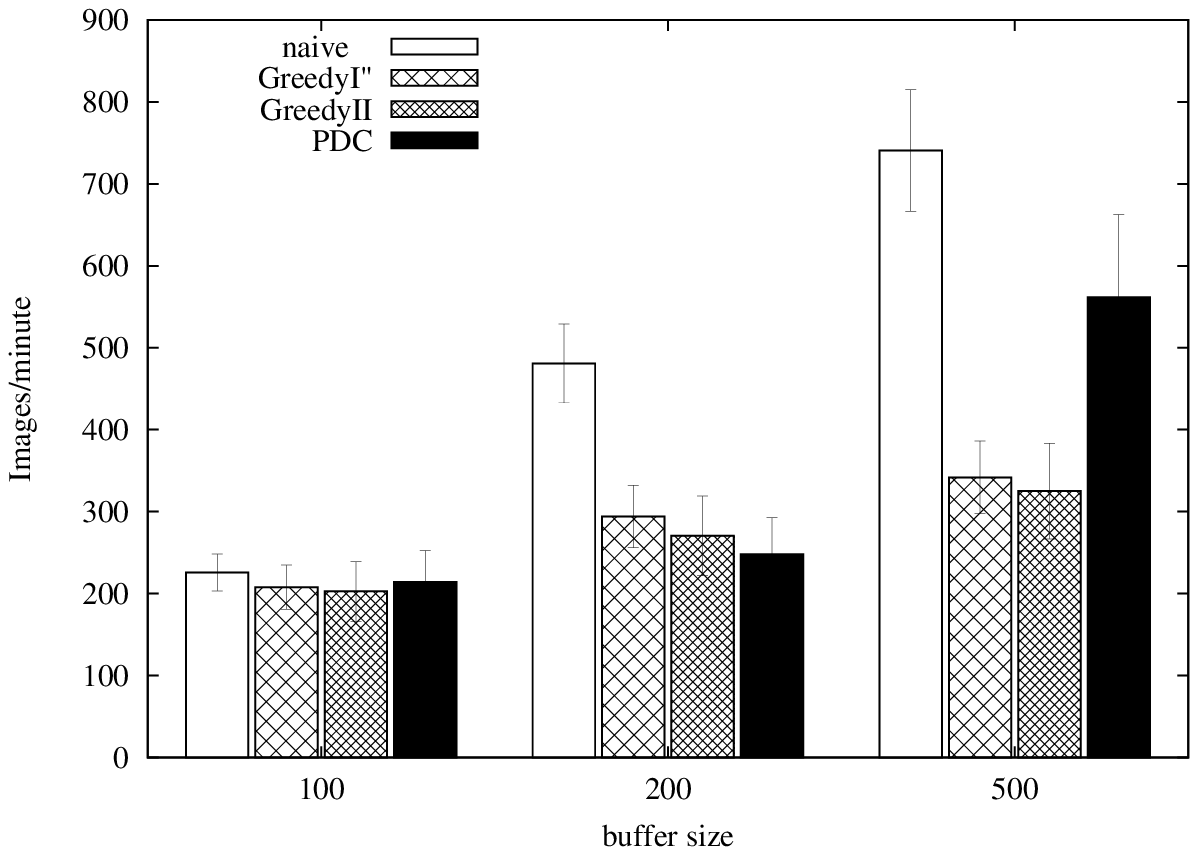}}
\subfigure[Signalling traffic\label{fig:over_signalling_buffer}]{\includegraphics[width=0.48\textwidth]{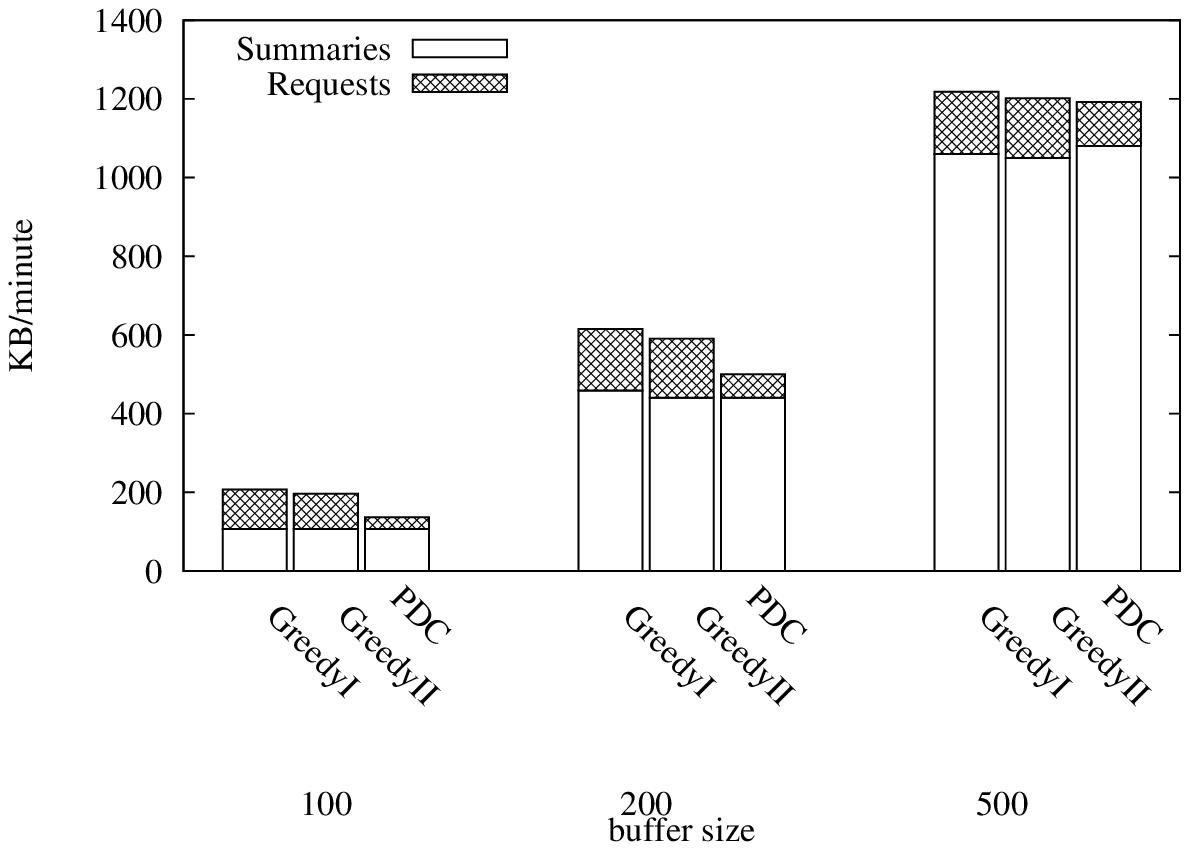}}
\caption{Bandwidth utilisation of signalling and data traffic versus the buffer size\label{fig:overhead_buffer}}
\vspace{-0.3cm}
\end{figure}
%

%
%
%
%
%
\subsection{Impact of latency requirements\label{sec:latency}}
\noindent
Clearly, the latency constraint has a significant impact on the system performance. Specifically, the shorter is the acceptable delay and the more images should be discarded because stale. Similarly, a stringent delay requirement can also make more difficult to obtain a complete coverage of the road network if the wireless roadside infrastructure is not sufficiently dense. To investigate the interplay between latency and system performance in this section we show results for a scenario in which 25~vehicles out of the 200~vehicles that are travelling in the road network are equipped with cameras, and we vary the tolerable latency. Specifically we investigate two scenarios with latency equal to 180 and 600 seconds, respectively (the case with latency equal to 300 is reported in Section~\ref{sec:density}). For the sake of figure clarity we show only results for $k \!=\! 100$.

Figure~\ref{fig:accuracy_latency} shows the cumulative distribution function (CDF) of the fraction of the total simulation time during which each road subsegment is covered by at least one image for different latencies. Our results indicate that there is a significant degradation of road network coverage for $T \!=\! 180$~seconds, and the differences between data collection schemes are less evident. It is not surprising that a more continuos coverage of road network is obtained when the delay requirement is less stringent (i.e., $T \!=\! 600$~seconds) as less images have to be discarded when a vehicle get associated to a roadside AP. Nevertheless, both \textsc{GreedyI} and \textsc{GreedyII} significantly outperforms the naive scheme also for $T \!=\! 600$~seconds because even if all stored images are valid most of them will carry redundant information. Thus, it is still essential to be able to select the images that may contribute the most to fill eventual gaps in the read network scene. The downside is that the large number of potentially redundant images that can be stored in the buffer for $T \!=\! 600$ also leads to an increase in the average number of images per road subsegment that are received by the data collector, as shown in Figure~\ref{fig:num_cell_latency}.
\begin{figure}[tb]
\centering
\subfigure[$T = 180$]{\includegraphics[width=0.4\textwidth]{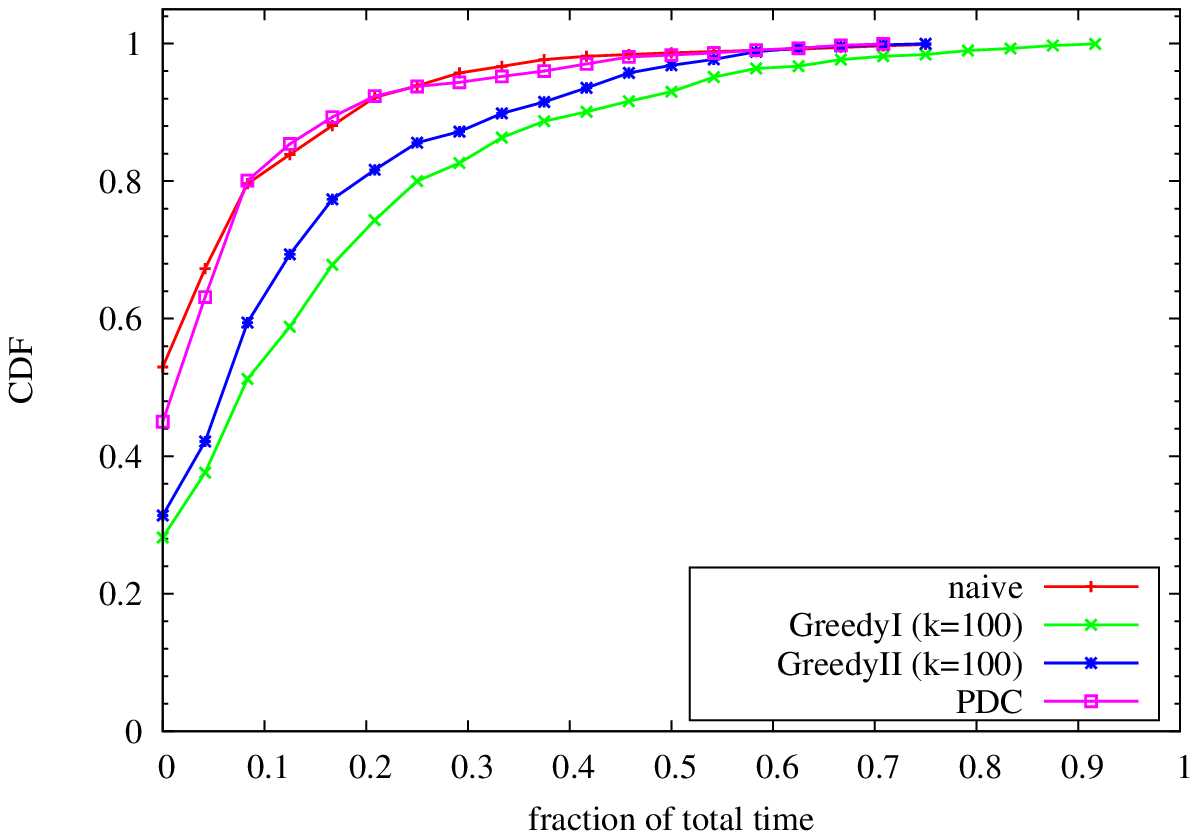}\label{fig:accuracy_latency180}}
\subfigure[$T = 600$]{\includegraphics[width=0.4\textwidth]{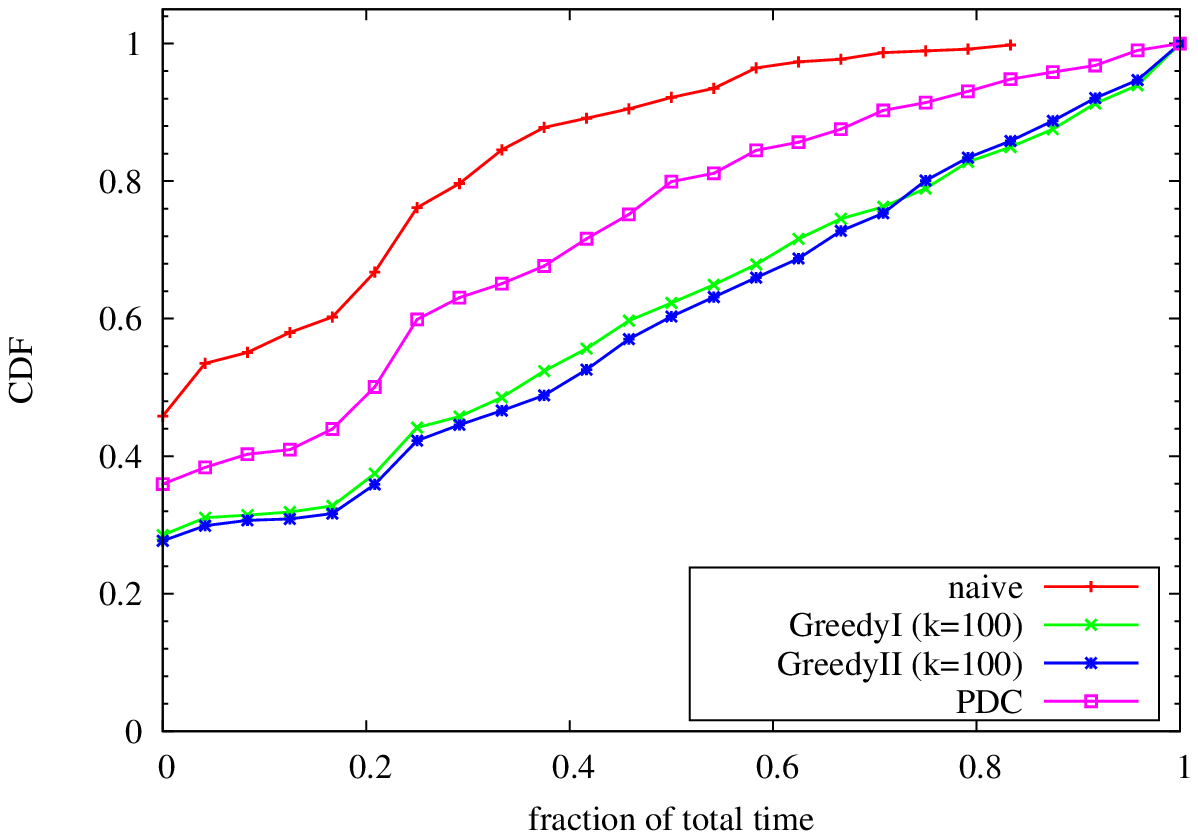}\label{fig:accuracy_latency600}}
\caption{CDF of the fraction of total time a road subsegment is covered by at least one image versus the latency requirement\label{fig:accuracy_latency}}
\vspace{-0.2cm}
\end{figure}
\begin{figure}[tb]
\centering
\includegraphics[width=0.5\textwidth,clip=true,angle=-90]{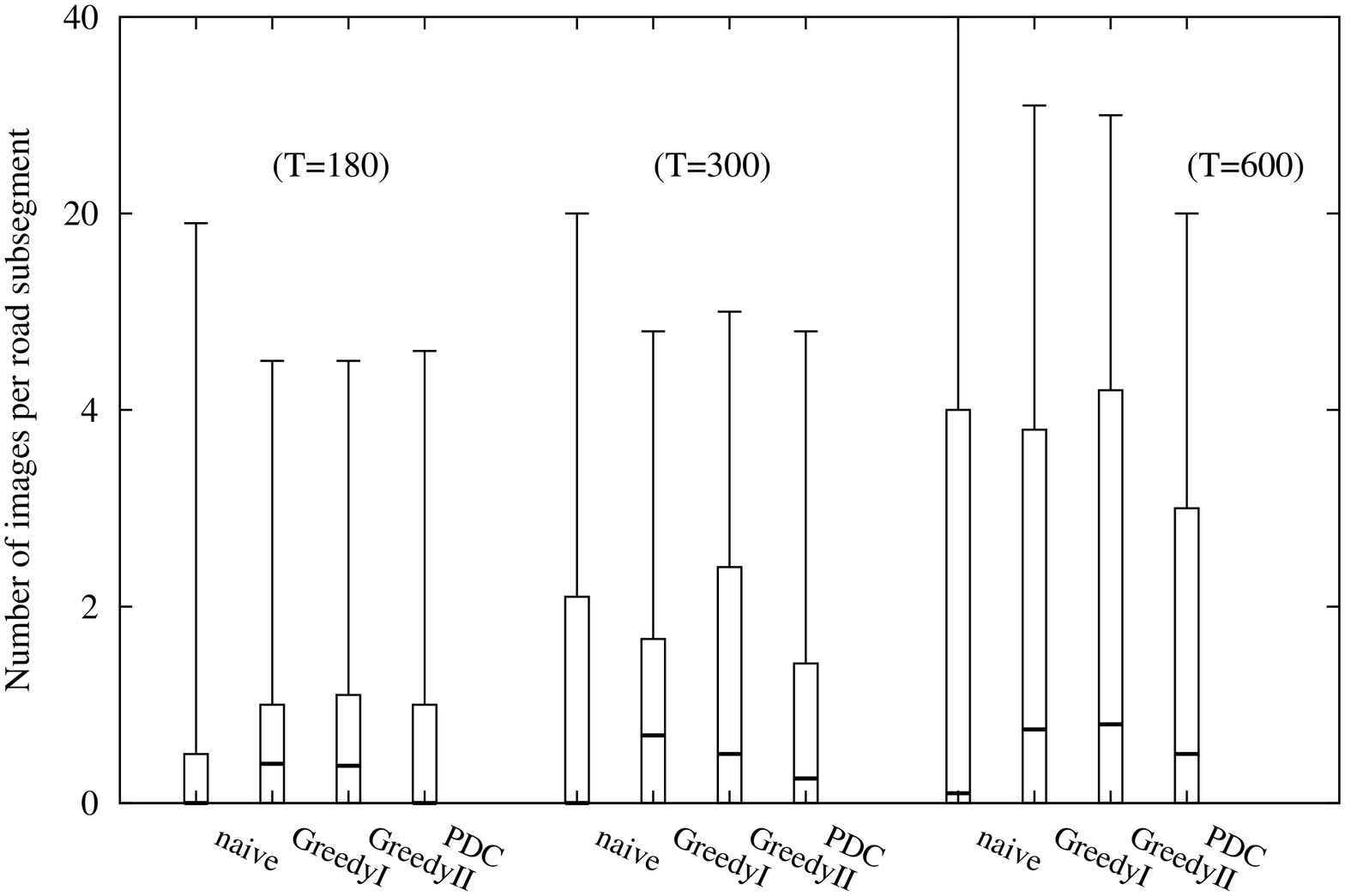}
\caption{Number of images per road subsegment versus the latency requirement}
\label{fig:num_cell_latency}
\end{figure}

Finally, Figure~\ref{fig:over_data_latency} shows the average number of images received at the data collector per minute for different latencies, while Figure~\ref{fig:over_signalling_latency} shows the protocol overheads in terms of sent data summaries and image requests. The figures confirm that one of the main effects of increasing the tolerable latency is that a larger number of images are still valid when a vehicle get associated to a roadside AP, which leads to an increase of the number of transferred images. 
\begin{figure}[tb]
\centering
\subfigure[Average number of images received at the data collector per minute\label{fig:over_data_latency}]{\includegraphics[width=0.48\textwidth]{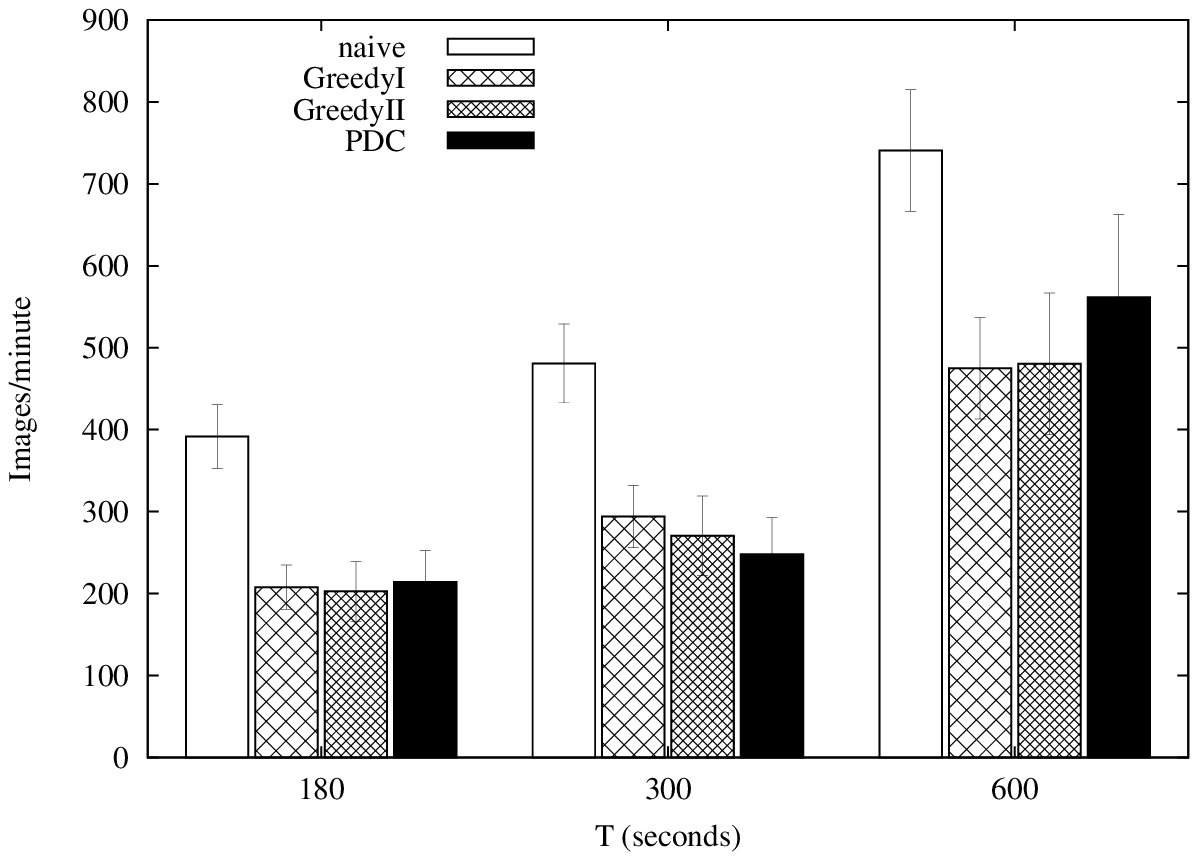}}
\subfigure[Signalling traffic\label{fig:over_signalling_latency}]{\includegraphics[width=0.48\textwidth]{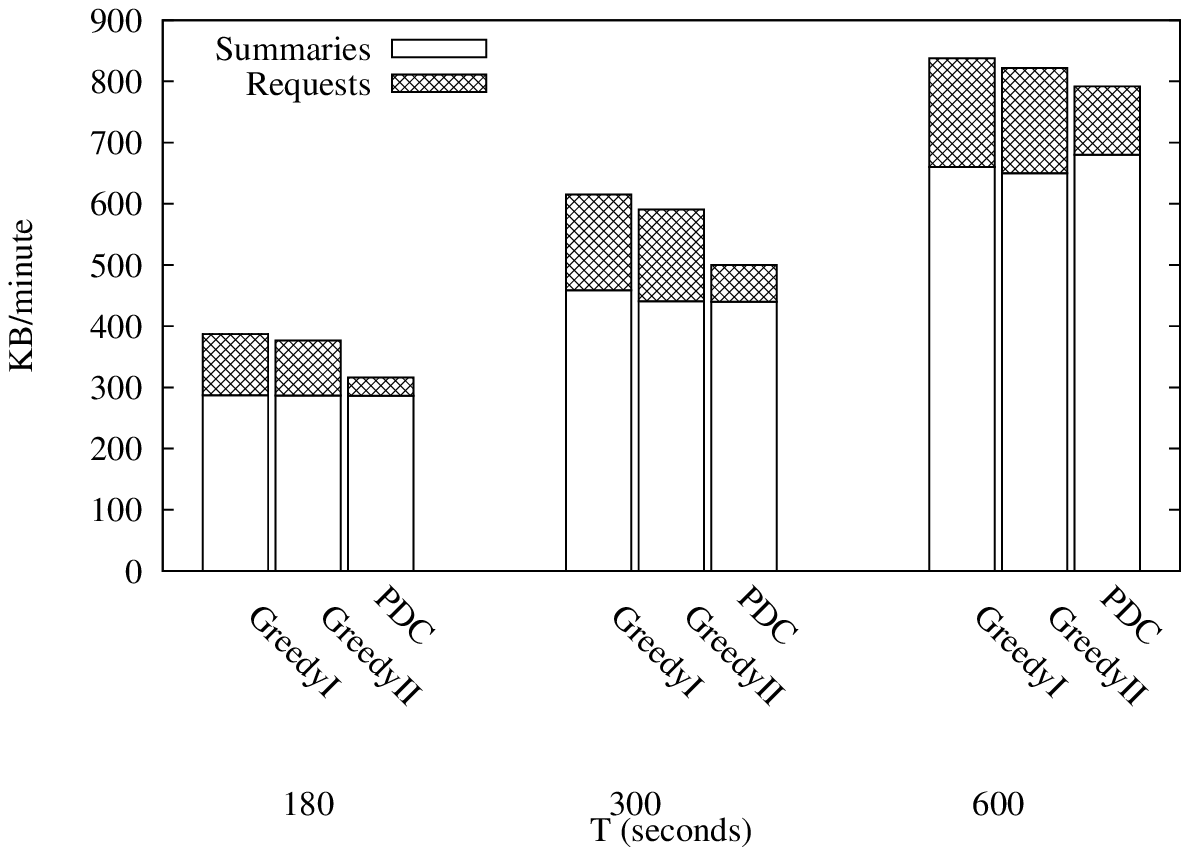}}
\caption{Bandwidth utilisation of signalling and data traffic versus the latency requirement\label{fig:overhead_latency}}
\vspace{-0.3cm}
\end{figure}
%
%
%
%
%
\subsection{Calibration of \textsc{PDC} parameters\label{sec:calibration}}
\noindent
The behaviour of \textsc{PDC} depends on three sets of parameters that determine: $i)$ how fast \textsc{PDC} updates the statistics on redundant images (i.e., $S$ parameter), $ii)$ how fast \textsc{PDC} explores different probabilities of requesting images of a given segment (i.e., $\delta p$), and $iii)$ which is the level of redundancy that should be accepted before changing the probability of requesting images of a given segment  (i.e., $r_{L}$ and $r_{H}$). Typically, the spatio-temporal distribution of received images is expected to change relatively slowly in normal conditions. Thus, it is enough to set $S < T$ as a reasonable trade-off between responsiveness and computational overheads. Similarly, a fine-grained exploration of the range $[0,1]$ does not provide a noticeable performance gain. On the contrary, if we select the $r_{H}$ parameter in a conservative manner we may be unable to increase the probability of requesting images. Similarly, if we select the $r_{L}$ parameter in an overoptimistic manner we may be unable to decrease the probability of requesting images. To clarify this concept in Figure~\ref{fig:pdc_calibration} we show the cumulative distribution function (CDF) of the fraction of the total simulation time during which each road subsegment is covered by at least one image for different \textsc{PDC} variants in the scenario with $n \!=\! 25$. The results indicate that the worse performance are obtained with low values of $r_{H}$ (i.e., $r_{H} \!=\! 3$) and when $r_{L}$ is close to $r_{H}$ (e.g., $r_{L} \!=\! 2$ and $r_{H} \!=\! 3$). This suggests that some degree of redundancy should be accepted to avoid that \textsc{PDC} discards too many images, including images covering road segments that are still uncovered. 
\begin{figure}[tb]
\centering
\includegraphics[width=0.7\textwidth,clip=true,angle=0]{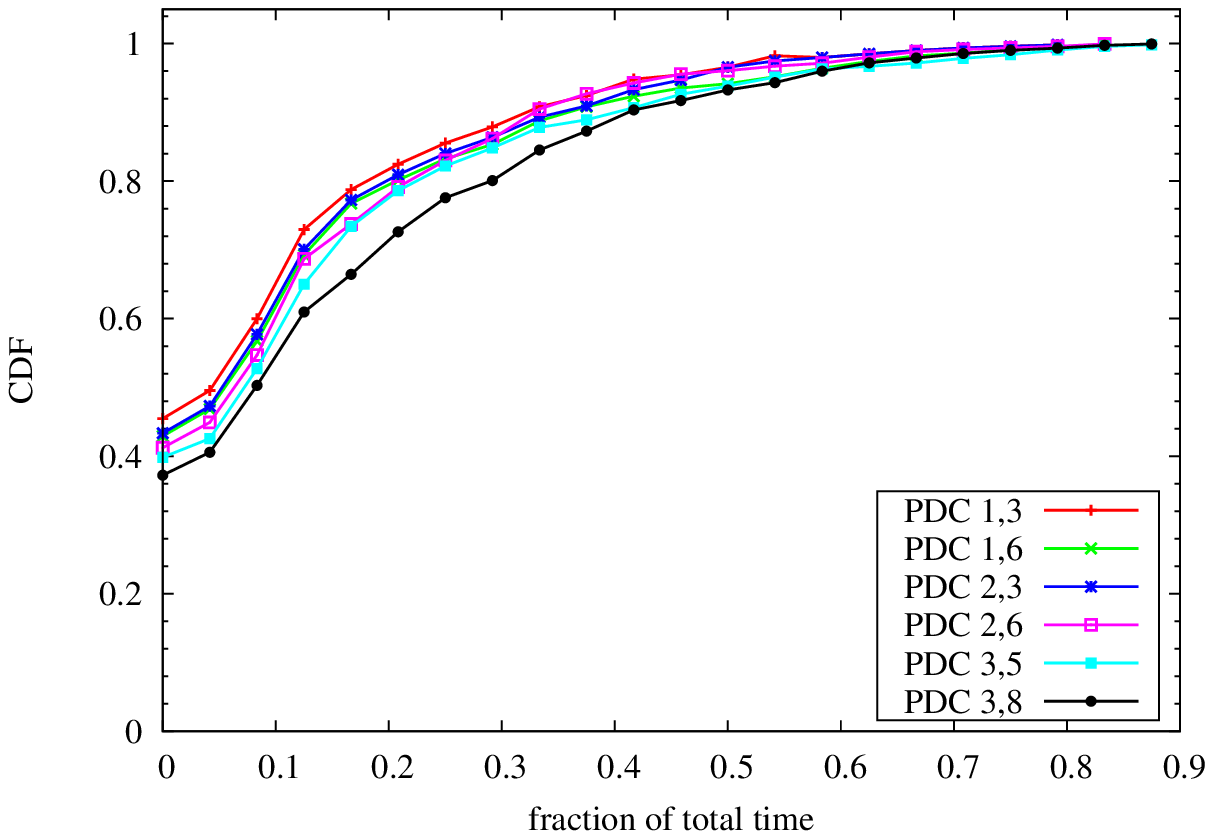}
\caption{Impact of $r_{L}$ and $r_{H}$ values on \textsc{PDC} performance for $n \!=\! 25$, $T \!=\! 300 seconds$ and $B \!=\! 200$~packets}
\label{fig:pdc_calibration}
\end{figure}
%

%
%
%
%
\section{Conclusions\label{sec:conclusions}}
\noindent
In this paper we have discussed which are the major challenges for leveraging on mobile cameras carried by vehicles to develop urban monitoring applications. In particular, we have pointed out that a critical issue is to minimise the number of camera snapshots that needs to be transferred from the vehicles without degrading the quality of the reconstructed road scene and violating the latency requirements that are imposed by the monitoring application. To tackle this problem we have formulated the data collection problem under network capacity constraints as a class of sub-modular set covering problems, whose solution can be approximated through efficient greedy heuristics. We also proposed a simpler scheme that operates in a more decentralized manner using basic aggregate information on the spatio-temporal distribution of received images. We carried out an in-depth performance comparison of those solutions using realistic vehicular mobility patterns. Results obtained using realistic vehicular mobility patterns in a wide range of different scenarios show that our data collection techniques ensure a more accurate coverage of the road network while significantly reducing the amount of transferred data. Finally, we note that the system  architecture proposed in this paper is amenable to a number of possible enhancements, such as more sophisticated buffer management techniques or the use of cooperative storage techniques and mechanisms for the suppression of replicated images.
%







\bibliographystyle{elsarticle-num}



\end{document}